\documentclass{amsart}

\copyrightinfo{2014}{Nicholas Coxon}
\usepackage{amssymb,url}

\newtheorem{theorem}{Theorem}[section]
\newtheorem{lemma}[theorem]{Lemma}
\newtheorem{corollary}[theorem]{Corollary}

\theoremstyle{definition}

\newtheorem{example}[theorem]{Example}

\theoremstyle{remark}
\newtheorem{remark}[theorem]{Remark}

\numberwithin{equation}{section}

\def\vec#1{\mathchoice{\mbox{\boldmath$\displaystyle#1$}}
{\mbox{\boldmath$\textstyle#1$}}
{\mbox{\boldmath$\scriptstyle#1$}}
{\mbox{\boldmath$\scriptscriptstyle#1$}}}
\newcommand{\norm}[1]{\left\|#1\right\|}
\newcommand{\Z}{\mathbb{Z}}
\newcommand{\R}{\mathbb{R}}
\newcommand{\id}{\mathrm{Id}}
\DeclareMathOperator{\res}{Res}
\DeclareMathOperator{\sres}{Sres}
\DeclareMathOperator{\syl}{Syl}
\DeclareMathOperator{\bez}{Bez}
\DeclareMathOperator{\vol}{vol}
\DeclareMathOperator{\adj}{adj} 
\DeclareMathOperator{\diag}{diag}
\DeclareMathOperator{\lc}{lc} 
\def\revdots{\mathinner{\mkern1mu\raise1pt\vbox{\kern7pt\hbox{.}}\mkern2mu\raise4pt\hbox{.}\mkern2mu \raise7pt\hbox{.}\mkern1mu}}

\begin{document}
%
\title[Montgomery's method of polynomial selection for NFS]{Montgomery's method of polynomial selection for the number field sieve}
%
\author{Nicholas Coxon}
\address{INRIA / CNRS / Université de Lorrain\'{e}, Campus Scientifique, BP 239, 54506 Vand\oe uvre-l\`{e}s-Nancy Cedex, France}
\email{nicholas.coxon@inria.fr}
%
\date{\today}
%
\begin{abstract} The number field sieve is the most efficient known algorithm for factoring large integers that are free of small prime factors. For the polynomial selection stage of the algorithm, Montgomery proposed a method of generating polynomials which relies on the construction of small modular geometric progressions. Montgomery's method is analysed in this paper and the existence of suitable geometric progressions is considered. 
\end{abstract}
\maketitle
%
	 		 
\section{Introduction}

In this paper, $N$ denotes a positive integer that is destined to be factored. When $N$ is large and free of small factors, the most efficient publicly known algorithm for determining its factors is the number field sieve~\cite{lenstra1993}. Such $N$ include RSA~\cite{rivest1978} moduli, for which numerous record factorisations have been achieved with the number field sieve, including the current 768-bit record~\cite{kleinjung2010}.

The number field sieve is comprised of several stages, commonly referred to as polynomial selection, sieving, filtering, linear algebra and square root computation. The polynomial selection stage requires the selection of coprime irreducible polynomials $f_{1},f_{2}\in\Z[x]$ that have a common root modulo $N$. After polynomial selection, sieving is used to identify coprime integer pairs $(a,b)$ such that the prime factors of $f_{i}(a/b)b^{\deg f_{i}}$ are below some bound $y_{i}$ for $i=1,2$. Obtaining sufficiently many pairs with this property, called \emph{relations}, is the most time consuming stage of the number field sieve, with the time taken greatly influenced by the choice of polynomials~\cite{murphy1998,murphy1999}.

Let $\Psi(x,y)$ denote the number of positive integers less than $x$ that are free of prime factors greater than $y$. Canfield, Erd\H{o}s and Pomerance~\cite{canfield1983} showed that for any $\varepsilon>0$, $\Psi(x,x^{1/u})=xu^{-u(1+o(1))}$ for $u\rightarrow\infty$, uniformly in the region $x\geq u^{u(1+\varepsilon)}$. It follows, heuristically, that in the polynomial selection stage of the number field sieve, the polynomials $f_{1}$ and $f_{2}$ should be chosen to minimise the size of the values $f_{1}(a/b)b^{\deg f_{1}}$ and $f_{2}(a/b)b^{\deg f_{2}}$ over the pairs $(a,b)$ considered in the sieve stage. Thus, it is necessary for the polynomials to have small coefficients. As a result, the degrees of $f_{1}$ and $f_{2}$ should not be too small. However, the degrees should not be too large either, since $f_{i}(a/b)b^{\deg f_{i}}$ is a homogeneous polynomial of degree $\deg f_{i}$ in $a$ and $b$. In practice, low-degree polynomials are used. For example, the two largest factorisations of RSA moduli~\cite{kleinjung2010,bai2012a} both used a sextic polynomial together with a linear polynomial. To quantify the coefficient size of a polynomial, the skewed $2$-norm $\norm{.}_{2,s}$ is used. The norm is defined as follows: if $f=\sum^{d}_{i=0}a_{i}x^{i}$ is a degree~$d$ polynomial with real coefficients, then 
\begin{equation*}
	\norm{f}_{2,s}%
	=\sqrt{\sum^{d}_{i=0}\left(a_{i}s^{i-\frac{d}{2}}\right)^{2}}%
	\quad\text{for all $s>0$}.
\end{equation*}
The parameter $s$ captures the shape of the sieve region, which is modelled by a rectangular region $[-A,A]\times(0,B]$ or an elliptic region
\begin{equation*}
	\left\{(x,y)\in\R^{2}\mid 0<y\leq B\sqrt{1-(x/A)^{2}}\right\}
\end{equation*}
such that $A/B=s$. In practice, the polynomial selection stage proceeds by first generating many ``raw'' polynomial  pairs with small coefficients. Then various methods of optimisation~\cite{murphy1999,bai2012,bai2014} are used to improve the quality of the raw pairs by taking into account additional factors that influence a pair's yield of relations, such as the presence of real roots and roots modulo small primes~\cite{murphy1998,murphy1999}.

The methods of polynomial selection used in all recent record factorisations~\cite{murphy1998,murphy1999,kleinjung2006,kleinjung2008} produce polynomials $f_{1}$ and $f_{2}$ such that one polynomial is linear. However, it is expected that a significant advantage is gained by using two nonlinear polynomials~\cite[Section~6.2.7]{crandall2005} (see also \cite[Section~4]{prest2012} for practical considerations relating to sieving). Montgomery~\cite{montgomery1993,montgomery2006} provided a method for generating two nonlinear polynomials with small coefficients. This paper extends and sharpens Montgomery's original analysis of the method.

\section{Montgomery's method}\label{sec:montgomerys-method}

A \emph{geometric progression} of \emph{length} $\ell$ and \emph{ratio} $r$ modulo $N$ is an integer vector $[c_{\ell-1},\ldots,c_{0}]$ such that $c_{i}\equiv c_{0}r^{i}\pmod{N}$ for $i=0,\ldots,\ell-1$. Square brackets are used to distinguish geometric progressions from regular vectors, which are denoted with round brackets. Montgomery~\cite{montgomery1993,montgomery2006} showed that a length $2d-1$ geometric progression $\vec{c}=[c_{2d-2},\ldots,c_{0}]$ modulo $N$ such that $\gcd(c_{0},c_{1},\ldots,c_{d-2},N)=1$ and
\begin{equation*}
	C=C(c_{2d-2},\ldots,c_{0})=\begin{pmatrix}
		c_{2d-2} & c_{2d-3} & \ldots & c_{d-1}\\
		c_{2d-3} & c_{2d-4} & \ldots & c_{d-2}\\
		\vdots & \vdots & & \vdots\\
		c_{d-1} & c_{d-2} & \ldots & c_{0}\\
	\end{pmatrix}
\end{equation*}
has full rank can be used to construct polynomials $f_{1},f_{2}\in\Z[x]$ of maximum degree $d$ that have a common root modulo $N$. Once a suitable geometric progression $\vec{c}$ has been found, Montgomery's method proceeds by computing a basis $\{(a_{1,d},\ldots,a_{1,0})^{T},(a_{2,d},\ldots,a_{2,0})^{T}\}$, where $A^{T}$ denotes the transpose of a matrix $A$, for the free $\Z$-module that is the set of integer vectors in the kernel of the matrix
\begin{equation}\label{eqn:partialC}
	\partial C=\begin{pmatrix}
		c_{2d-2} & c_{2d-3} & \ldots & c_{d-2}\\
		c_{2d-3} & c_{2d-4} & \ldots & c_{d-3}\\
		\vdots & \vdots & & \vdots\\
		c_{d} & c_{d-1} & \ldots & c_{0}\\
	\end{pmatrix}.
\end{equation}
The basis yields polynomials $f_{1}=\sum^{d}_{i=0}a_{1,i}x^{i}$ and $f_{2}=\sum^{d}_{i=0}a_{2,i}x^{i}$. If $r$ is the ratio of the geometric progression modulo $N$, then $(r^{d},r^{d-1},\ldots,1)$ is a linear combination of the row vectors of $\partial C$ modulo $N$ since $\gcd(c_{0},c_{1},\ldots,c_{d-2},N)=1$. Thus, $r$ is a root of $f_{1}$ and $f_{2}$ modulo $N$. Denote by $\widehat{\partial C}$ the submatrix of $\partial C$ obtained by deleting its first column. Then $\widehat{\partial C}$ has full rank since it is equal to the submatrix of $C$ obtained by deleting its first row. Consequently, the kernel of $\partial C$ is $2$-dimensional. Moreover, and at least one of $f_{1}$ and $f_{2}$ has degree equal to $d$, otherwise $(a_{1,d-1},\ldots,a_{1,0})^{T}$ and $(a_{2,d-1},\ldots,a_{2,0})^{T}$ are linearly independent and in the $1$-dimensional kernel of $\widehat{\partial C}$, which is absurd. Finally, as was observed by Montgomery~\cite{montgomery1993} and which is shown to hold in this paper, the polynomials $f_{1}$ and $f_{2}$ are coprime since $C$ is nonsingular.

To ensure that the norms $\norm{f_{1}}_{2,s}$ and $\norm{f_{2}}_{2,s}$ are small for some $s>0$, the basis is chosen such that
\begin{equation}\label{eqn:skewed-kernel-basis}
	\left\{\left(a_{1,d}s^{d},a_{1,d-1}s^{d-1},\ldots,a_{1,0}\right)^{T},\left(a_{2,d}s^{d},a_{2,d-1}s^{d-1},\ldots,a_{2,0}\right)^{T}\right\}
\end{equation}
is Lagrange-reduced (see \cite[p.~41]{nguyen2010}). As a result,
\begin{equation}\label{eqn:poly-gp-bnd}
	s^{\frac{\deg f_{1}+\deg f_{2}}{2}-d}\norm{f_{1}}_{2,s}\norm{f_{2}}_{2,s}\leq\frac{\gamma_{2}}{N^{d-2}}\norm{\vec{c}}^{d-1}_{2,s^{-1}}
\end{equation}
(see \cite[Section~3.2]{coxon2011}), where $\gamma_{2}=2/\sqrt{3}$ is Hermite's constant for dimension two, and the vector norm $\norm{.}_{2,s}$ is defined as follows: if $\vec{v}=(v_{n},v_{n-1},\ldots,v_{0})$ is a real $(n+1)$-dimensional vector, then
\begin{equation*}
	\norm{\vec{v}}_{2,s}%
	=\sqrt{\sum^{n}_{i=0}\left(v_{i}s^{i-\frac{n}{2}}\right)^{2}}
	\quad\text{for all $s>0$}.
\end{equation*}
Consequently, it is a requirement of Montgomery's method that $\norm{\vec{c}}_{2,s^{-1}}$ is small. Therein lies the difficulty of the method, as the basis \eqref{eqn:skewed-kernel-basis} is readily computed in polynomial time (see \cite[Section~3.1.2]{coxon2011}). The problem of constructing small geometric progressions has been addressed by several authors~\cite{montgomery1993,williams2010,prest2012,koo2011,coxon2011}.

It is natural to consider the existence of small geometric progressions. Montgomery~\cite{montgomery1993} showed that if there exist two degree $d$ polynomials that have small coefficients and a common root $r$ modulo $N$, then there exists a small length $2d-1$ geometric progression with ratio $r$ modulo $N$. Montgomery's proof is constructive and is generalised to two polynomials $f_{1}$ and $f_{2}$ of maximum degree $d$ in this paper. Furthermore, it is shown that if the polynomials are coprime, then the geometric progression given by the construction is the unique vector $(c_{2d-1},\ldots,c_{0})$, up to scalar multiple, such that the coefficient vectors of $f_{1}$ and $f_{2}$ are in the kernel of the matrix $\partial C$ defined in \eqref{eqn:partialC}. As a result, the analysis of the construction contributes to the analysis of Montgomery's method.

This paper is organised as follows: the definitions and some properties of the Sylvester matrix, the Bezout matrix and the resultant are reviewed in the next section; the generalisation of Montgomery's geometric progression construction is presented and analysed in Section~\ref{sec:existence}; and a full analysis of Montgomery's method is provided in Section~\ref{sec:analysis}.

\section{The Sylvester matrix, the Bezout matrix and the resultant}

Matrices with the property that each of their rows contain the coefficients of some polynomial are frequently encountered in this paper. The Sylvester and Bezout matrices are constructed in this manner. Consequently, compact notation for such matrices is defined before introducing the protagonists of this section.

For $m\geq 1$ polynomials $f_{1},\ldots,f_{m}$ and any integer $n\geq\max_{1\leq i\leq m}\deg f_{i}$, denote by $(f_{1},\ldots,f_{m})_{n}$ the $m\times(n+1)$ matrix $(a_{i,j})_{i=1,\ldots,m; j=1,\ldots,n+1}$ where $a_{i,j}$ is the coefficient of $x^{n+1-j}$ in $f_{i}$. When $n=\max_{1\leq i\leq m}\deg f_{i}$, the subscript $n$ is drop, giving the notation $(f_{1},\ldots,f_{m})$. The parameter $n$ is viewed as the formal degree of the polynomials $f_{1},\ldots,f_{m}$. For example, $(f)$ is the vector of coefficients of $f$, while $(f)_{n}$ for some $n\geq\deg f$ is the vector of coefficients of $f$ when view as a polynomial of formal degree $n$. Define $\left(\ \right)_{n}$ to be the $0\times n$ empty matrix.

\subsection{The Sylvester matrix and the resultant}

Let $\mathbb{A}$ be an integral domain. Then the \emph{Sylvester matrix} of non-constant polynomials $f_{1},f_{2}\in\mathbb{A}[x]$ is the matrix
\begin{equation*}
  \syl(f_{1},f_{2})=\left(x^{\deg f_{2}-1}f_{1},\ldots,f_{1},x^{\deg f_{1}-1}f_{2},\ldots,f_{2}\right).
\end{equation*}
The determinant of $\syl(f_{1},f_{2})$ is called the \emph{resultant} of $f_{1}$ and $f_{2}$, and is denoted $\res(f_{1},f_{2})$. The resultant of $f_{1}$ and $f_{2}$ is zero if and only if the polynomials have a nontrivial gcd over the field of fractions of $\mathbb{A}$.

For non-constant polynomials $f_{1},f_{2}\in\Z[x]$ and all real numbers $s>0$, define $\theta_{s}(f_{1},f_{2})$ to be the angle (in $[0,\pi]$) between the row vectors of $(f_{1}(sx),f_{2}(sx))$. The following lemma provides upper and lower bounds on the resultant of a pair of number field sieve polynomials (see \cite[Section~2.1.2]{coxon2011} for a proof):

\begin{lemma}\label{lem:resultantBound} Suppose that $f_{1},f_{2}\in\Z[x]$ are non-constant, coprime and have a common root modulo $N$. Then
\begin{equation*}
  N\leq\left|\res(f_{1},f_{2})\right|\leq\left|\sin\theta_{s}(f_{1},f_{2})\right|^{\min\{\deg f_{1},\deg f_{2}\}}\norm{f_{1}}^{\deg f_{2}}_{2,s}\norm{f_{2}}^{\deg f_{1}}_{2,s}
\end{equation*}
for all $s>0$.
\end{lemma}

As a consequence of Lemma~\ref{lem:resultantBound}, a pair of number field sieve polynomials in considered to have optimal resultant if it is equal to $N$, and optimal coefficient size if $\norm{f_{1}}^{\deg f_{2}}_{2,s}\norm{f_{2}}^{\deg f_{1}}_{2,s}$ is $O(N)$. Thus, inequality \eqref{eqn:poly-gp-bnd} implies that a pair of degree $d$ polynomials generated by Montgomery's method has optimal coefficient size if $\norm{\vec{c}}_{2,s^{-1}}=O(N^{1-1/d})$.

\subsection{The Bezout matrix}

Let $\mathbb{A}$ be an integral domain and $f_{1},f_{2}\in\mathbb{A}[x]$ be non-constant. Write $f_{i}=\sum^{d}_{j=0}a_{i,j}x^{j}$ for $i=1,2$ such that $d=\max\{\deg f_{1},\deg f_{2}\}$ and the coefficients $a_{i,j}$ are elements of $\mathbb{A}$. Define polynomials
\begin{equation}\label{eqn:bezout-pols}
	p_{i+1}=\left(\sum^{i}_{j=0}a_{2,d-i+j}x^{j}\right)f_{1}-\left(\sum^{i}_{j=0}a_{1,d-i+j}x^{j}\right)f_{2}\quad\text{for $i=0,\ldots,d-1$}.
\end{equation}
Then the \emph{Bezout matrix}, or \emph{Bezoutian}, of $f_{1}$ and $f_{2}$ is the matrix $\mathrm{Bez}(f_{1},f_{2})=(p_{1},\ldots,p_{d})_{d-1}$. Denote by $\lc(f)$ the leading coefficient of a polynomial $f$. Similar to the Sylvester matrix, the determinant of the Bezout matrix of $f_{1}$ and $f_{2}$ is related to their resultant (see \cite[Section~2]{sederberg97}):
\begin{equation}\label{eqn:det-bezout}
  \det\mathrm{Bez}(f_{1},f_{2})=(-1)^{\frac{d(d+1)}{2}}\lc(f_{1})^{d-\deg f_{2}}\left((-1)^{d}\lc(f_{2})\right)^{d-\deg f_{1}}\res(f_{1},f_{2}).
\end{equation}

A \emph{Hankel matrix} over $\mathbb{A}$ is a square matrix $H=(h_{i,j})_{i=1,\ldots,n;j=1,\ldots,n}$ such that $h_{i,j}=h_{i+j-1}$ for some $h_{1},\ldots,h_{2n-1}\in\mathbb{A}$. Lander~\cite{lander1974} showed that the inverse of a Bezout matrix is a Hankel matrix and, conversely, that the inverse of a Hankel matrix is the Bezout matrix of two polynomials. For an $m\times n$ matrix $H=(h_{i,j})_{i=1,\ldots,m;j=1,\ldots,n}$ such that $h_{i,j}=h_{i+j-1}$ for some $h_{1},\ldots,h_{m+n-1}\in\mathbb{A}$, define $\partial^{k}H=(h_{i+j-1})_{i=1,\ldots,m+k;j=1,\ldots,n-k}$ for $k=0,\ldots,m-1$. Let $\partial H$ denote the matrix $\partial^{1}H$. Define $\widehat{\partial H}=(h_{i+j})_{i=1,\ldots,m+1;j=1,\ldots,n-2}$. The following result of Heinig and Rost~\cite[Theorem~4.2]{heinig2010} expresses the inverse of a real Hankel matrix $H$ as the Bezout matrix of two polynomials obtained from the kernel of $\partial H$:
\begin{lemma}\label{lem:HankelInverse} Let $H=(h_{i+j-1})_{i=1,\ldots,d;j=1,\ldots,d}$ be a real nonsingular Hankel matrix. If $f_{1},f_{2}\in\R[x]$ such that $\left\{(f_{1})^{T}_{d},(f_{2})^{T}_{d}\right\}$ is a basis of the kernel of $\partial H$, then
\begin{equation*}
	H^{-1} = -\frac{1}{\det\psi}\bez(f_{1},f_{2})
	\quad\text{where}\quad
	\psi=\begin{pmatrix}
  		h_{d} & \ldots & h_{2d-1} & 0\\
  		0 & \ldots & 0 & 1
	\end{pmatrix}%
	\left(f_{1},f_{2}\right)^{T}_{d}.
\end{equation*}
\end{lemma}
The formulation of the Bezout matrix presented in this paper is different to that used by Heinig and Rost. Thus, it is necessary to convert between the two formulations:
\begin{proof}[Proof of Lemma~\ref{lem:HankelInverse}] Let $H=(h_{i+j-1})_{i=1,\ldots,d;j=1,\ldots,d}$ be a real nonsingular Hankel matrix. Suppose that $f_{1},f_{2}\in\R[x]$ such that $\left\{(f_{1})^{T}_{d},(f_{2})^{T}_{d}\right\}$ is a basis of the kernel of $\partial H$. Define $\psi$ as in the statement of the theorem. The matrix $\widehat{\partial H}$ is equal to the submatrix of $H$ obtained by deleting its first row. Thus, $\widehat{\partial H}$ has full rank. It follows that $\max\{\deg f_{1},\deg f_{2}\}=d$, otherwise $(f_{1})^{T}_{d-1}$ and $(f_{2})^{T}_{d-1}$ would be linearly independent and in the kernel of $\widehat{\partial H}$. Interchanging $f_{1}$ and $f_{2}$ changes $\bez(f_{1},f_{2})$ and $\det\psi$ by a factor of $-1$. Therefore, assume without loss of generality that $\deg f_{1}=d$. Write $f_{i}=\sum^{d}_{j=0}a_{i,j}x^{j}$ for $i=1,2$ such that the coefficients $a_{i,j}$ are real numbers. Define $g_{i}=\sum^{d}_{j=0}a_{i,d-j}x^{j}$ for $i=1,2$. Then Theorem~4.2 of Heinig and Rost~\cite{heinig2010} implies that
\begin{equation}\label{eqn:H-inverse}
	H^{-1} = -\frac{1}{\det\psi}\mathrm{B}(g_{1},g_{2})
\end{equation}
where the matrix $\mathrm{B}(g_{1},g_{2})$, called the (Hankel) Bezoutian of $g_{1}$ and $g_{2}$ by Heinig and Rost~\cite[Section~2.1]{heinig2010} (and many other authors), is defined as follows: $\mathrm{B}(g_{1},g_{2})=\left(b_{i,j}\right)_{i=1,\ldots,d;j=1,\ldots,d}$ such that $b_{i,j}$ is the coefficient of $x^{i-1}y^{i-1}$ in the polynomial
\begin{equation*}
	b(x,y)=\frac{g_{1}(x)g_{2}(y)-g_{2}(x)g_{1}(y)}{x-y}.
\end{equation*}

Expanding the numerator of $b(x,y)$ shows that
\begin{align*}
	b(x,y)%
	&=\sum^{d}_{k=0}a_{2,d-k}\left(\sum^{d}_{i=0}a_{1,d-i}\frac{x^{i}y^{k}-x^{k}y^{i}}{x-y}\right)\\
	&=\sum^{d}_{k=0}a_{2,d-k}\left(\sum^{d}_{i=k+1}a_{1,d-i}\frac{x^{i-k}-y^{i-k}}{x-y}x^{k}y^{k}-\sum^{k-1}_{i=0}a_{1,d-i}\frac{x^{k-i}-y^{k-i}}{x-y}x^{i}y^{i}\right).
\end{align*}
Therefore,
\begin{equation*}\label{eqn:bezH-revpol}
	\mathrm{B}(g_{1},g_{2})%
	=\sum^{d}_{k=0}a_{2,d-k}%
	\begin{pmatrix}
		& & -a_{1,d} & & &\\
		& \revdots & \vdots & & &\\
		-a_{1,d} & \ldots & -a_{d-k+1} & & &\\
		& & & a_{d-k-1} & \ldots & a_{0}\\
		& & & \vdots & \revdots & \\
		& & & a_{0} &  &\\
	\end{pmatrix},
\end{equation*}
where the omitted entries are zeros. Let $k\in\Z$ such that $0\leq k\leq d$, and $p_{1},\ldots,p_{d}\in\R[x]$ such that $\bez(f_{1},x^{d-k})=(p_{1},\ldots,p_{d})$. Then
\begin{equation*}
	p_{i+1}=-x^{d-k}\sum^{i}_{j=0}a_{1,d-i+j}x^{j}%
	\quad\text{for $i=0,\ldots,k-1$},
\end{equation*}
and
\begin{equation*}
	p_{i+1}=x^{i-k}f_{1}-x^{k-d}\sum^{i}_{j=0}a_{1,d-i+j}x^{j}%
	=x^{i-k}\sum^{d-i-1}_{j=0}a_{j}x^{j}%
	\quad\text{for $i=k,\ldots,d-1$}.
\end{equation*}
Hence,
\begin{equation*}
	\mathrm{B}(g_{1},g_{2})%
	=\sum^{d}_{k=0}a_{2,d-k}\bez(f_{1},x^{d-k})%
	=\bez(f_{1},f_{2}).
\end{equation*}
Combining this equation with \eqref{eqn:H-inverse} completes the proof.
\end{proof}

\section{Existence of geometric progressions}\label{sec:existence}

The \emph{integer kernel} of an integer matrix $A$ is the free $\Z$-module consisting of all integer column vectors $\vec{x}$ such that $A\vec{x}=\vec{0}$. Montgomery~\cite{montgomery1993} showed that if two coprime degree $d\geq 2$ integer polynomials $f_{1}$ and $f_{2}$ have a common root modulo~$N$, then a vector $\vec{c}\in\Z^{2d-1}$ such that $\vec{c}^{T}$ spans the integer kernel of the matrix
\begin{equation}\label{eqn:Sd}
	\left(x^{d-2}f_{1},\ldots,f_{1},x^{d-2}f_{2},\ldots,f_{2}\right)
\end{equation}
is a geometric progression modulo $N$. Moreover, $\norm{\vec{c}}_{2,1}=O(\norm{f_{1}}^{d-1}_{2,1}\norm{f_{2}}^{d-1}_{2,1})$ since the entries of $\vec{c}$ are, up to a constant, order $2d-2$ minors of the matrix in \eqref{eqn:Sd}. Koo, Jo and Kwon~\cite[Theorem~2]{koo2011} generalise Montgomery's result by providing a construction which, for $k\in\{1,\ldots,d-1\}$, uses $j=\lceil(d-1)/k\rceil+1$ degree $d$ polynomials $f_{1},\ldots,f_{j}$ with a common root modulo $N$ to construct a length $d+k-1$ geometric progression $\vec{c}$ such that $\norm{\vec{c}}_{2,1}=O(\max_{1\leq i\leq j}\norm{f_{i}}^{d+k-1}_{2,1})$. In this section, another generalisation is presented, with two polynomials $f_{1}$ and $f_{2}$ of maximum degree $d$ used to construct geometric progressions of lengths $\deg f_{1}+\deg f_{2}-1,\ldots,2d-1$. The geometric progressions are shown to have small size whenever $f_{1}$ and $f_{2}$ have small coefficients. Therefore, the existence of small geometric progressions for Montgomery's method, and relaxations of the method that employ shorter progressions~\cite{williams2010,prest2012,koo2011,coxon2011}, is established under the assumption that good nonlinear polynomial pairs exist.

Let $\mathbb{A}$ be an integral domain. For $f_{1},f_{2}\in\mathbb{A}[x]$ such that $2\leq\deg f_{2}\leq\deg f_{1}$, and $t\in\{\deg f_{2},\ldots,\deg f_{1}\}$, define the $(\deg f_{1}+t-2)\times(\deg f_{1}+t-1)$ matrix
\begin{equation*}
	\mathrm{S}_{t}(f_{1},f_{2})%
	=\left(x^{t-2}f_{1},\ldots,f_{1},x^{\deg f_{1}-2}f_{2},\ldots,f_{2}\right).
\end{equation*}
Define signed minors $\mathrm{M}_{t,1}(f_{1},f_{2}),\ldots,\mathrm{M}_{t,\deg f_{1}+t-1}(f_{1},f_{2})$ of the matrix $\mathrm{S}_{t}(f_{1},f_{2})$ as follows: for $i=1,\ldots,\deg f_{1}+t-1$, $\mathrm{M}_{t,i}(f_{1},f_{2})$ is equal to $(-1)^{1+i}$ times the determinant of the submatrix of $\mathrm{S}_{t}(f_{1},f_{2})$ obtained by deleting its $i$th column. Define
\begin{equation*}
  \vec{c}_{t}(f_{1},f_{2})%
  =\left(\mathrm{M}_{t,1}(f_{1},f_{2}),\ldots,\mathrm{M}_{t,\deg f_{1}+t-1}(f_{1},f_{2})\right)
\end{equation*}
for $t=\deg f_{2},\ldots,\deg f_{1}$. When $f_{1}$ and $f_{2}$ are clear from the context, $\mathrm{S}_{t}$ is used to denote $\mathrm{S}_{t}(f_{1},f_{2})$, $\mathrm{M}_{t,i}$ is used to denote $\mathrm{M}_{t,i}(f_{1},f_{2})$, and $\vec{c}_{t}$ is used to denote $\vec{c}_{t}(f_{1},f_{2})$. The $i$th entry of the vector $\mathrm{S}_{t}(f_{1},f_{2})\cdot\vec{c}_{t}(f_{1},f_{2})^{T}$ is equal to the determinant of the matrix obtained by appending the $i$th row of $\mathrm{S}_{t}(f_{1},f_{2})$ to its top, and thus is equal to zero. Therefore,
\begin{equation}\label{eqn:ct-in-kernel}
	\mathrm{S}_{t}(f_{1},f_{2})\cdot\vec{c}_{t}(f_{1},f_{2})^{T}=\vec{0}_{\deg f_{1}+t-2}\quad\text{for $t=\deg f_{2},\ldots,\deg f_{1}$},
\end{equation}
where $\vec{0}_{n}$ denotes the $n$-dimensional column vector of zeros for $n=1,2,\ldots$. If $A$ is an $m\times n$ integer matrix such that $m\leq n$, define $\Delta(A)$ to be the greatest common divisor of all $m\times m$ minors of $A$, with $\Delta(A)=0$ if all such minors are zero. Then $\Delta(A)$ is nonzero if and only if $A$ has full rank.

Let $f_{1},f_{2}\in\Z[x]$ be coprime degree $d\geq 2$ polynomials that have a common root modulo $N$. Then $\mathrm{S}_{d}(f_{1},f_{2})$ is the matrix in \eqref{eqn:Sd}. Therefore, $\mathrm{S}_{d}(f_{1},f_{2})$ is the submatrix of $\syl(f_{1},f_{2})$ obtained by first deleting its first and $(d+1)$th rows, giving a matrix whose first column contains zeros, then deleting the first column of the resulting matrix. Thus, $\mathrm{S}_{d}(f_{1},f_{2})$ has full rank, otherwise $\syl(f_{1},f_{2})$ is singular. Therefore, $\vec{c}_{d}(f_{1},f_{2})^{T}$ is nonzero and in the integer kernel of $\mathrm{S}_{d}(f_{1},f_{2})$. Consequently, Montgomery's result implies $\vec{c}_{d}(f_{1},f_{2})/\Delta(\vec{c}_{d}(f_{1},f_{2}))$ is a geometric progression modulo $N$. More generally, if $2\leq\deg f_{2}\leq\deg f_{1}$, then the vectors $\vec{c}_{t}(f_{1},f_{2})$ for $t=\deg f_{2},\ldots,\deg f_{1}$ are geometric progressions modulo $N$:

\begin{theorem}\label{thm:polys-to-gp} Let $f_{1}$ and $f_{2}$ be coprime integer polynomials such that $2\leq\deg f_{2}\leq\deg f_{1}$, $f_{1}$ and $f_{2}$ have a common root $r\in\Z$ modulo $N$, and $\lc(f_{1})\,\Delta(\mathrm{S}_{\deg f_{2}}(f_{1},f_{2}))$ is relatively prime to $N$. Then, for $t\in\{\deg f_{2},\ldots,\deg f_{1}\}$, the vector
\begin{equation*}
  \vec{c}_{t}%
  =\vec{c}_{t}(f_{1},f_{2})%
  =\left(c_{t,\deg f_{1}+t-2},\ldots,c_{t,0}\right)
\end{equation*}
satisfies the following properties:
\begin{enumerate}
	\item\label{isGP} $\vec{c}_{t}$ is a nonzero geometric progression with ratio $r$ modulo $N$;
	\item\label{gcdCondition} $\gcd(c_{t,0},N)=1$;
	\item\label{rankCondition} the matrix $C_{t}=\left(c_{t,\deg f_{1}+t-i-j}\right)_{i=1,\ldots,t;j=1,\ldots,\deg f_{1}}$ has full rank;
	\item\label{kernelCondition} the vectors $(f_{1})^{T}_{\deg f_{1}}$ and $(f_{2})^{T}_{\deg f_{1}}$ are in the kernel of $\partial C_{t}$; and
	\item\label{GPsizeBounds} the inequalities
	\begin{equation*}
  		\norm{\vec{c}_{t}}_{2,s^{-1}}\leq\left(\left|\sin\theta_{s}(f_{1},f_{2})\right|\norm{f_{1}}_{2,s}\right)^{t-1}\left(s^{\frac{\deg f_{2}-t}{2}}\norm{f_{2}}_{2,s}\right)^{\deg f_{1}-1}
\end{equation*}
	and
	\begin{equation*}
  		\norm{\vec{c}_{t}}_{2,s^{-1}}%
	\geq s^{\frac{t-\deg f_{1}}{2}}\left|s^{\deg f_{2}}\lc(f_{2})\right|^{\frac{\deg f_{1}}{t}-1}\left|\lc(f_{1})^{t-\deg f_{2}}\res(f_{1},f_{2})\right|^{1-\frac{1}{t}}
\end{equation*}
hold for all $s>0$.
\end{enumerate}
\end{theorem}

Theorem~\ref{thm:polys-to-gp} is proved as a series of lemmas and corollaries in the next section, with some of the properties presented as part of more general statements. The requirement in the theorem that $\Delta(\mathrm{S}_{\deg f_{2}}(f_{1},f_{2}))$ and $N$ are relatively prime has not been encountered so far in this paper, nor is it usual to require a pair of nonlinear number field sieve polynomials to satisfy it. In the next section, it is shown that
\begin{equation*}
	\Delta(\mathrm{S}_{\deg f_{2}}(f_{1},f_{2}))\leq\left(\left|\sin\theta_{s}(f_{1},f_{2})\right|\norm{f_{1}}_{2,s}\right)^{\deg f_{2}-1}\norm{f_{2}}^{\deg f_{1}-1}_{2,s}\quad\text{for all $s>0$}
\end{equation*}
(see Lemma~\ref{lem:ct-upper-bnd} and Remark~\ref{rmk:ct-upper-bnd}). Therefore, if $f_{1}$ and $f_{2}$ are a pair of number field sieve polynomials that are close the attaining the lower bound from Lemma~\ref{lem:resultantBound}, then $\gcd(\Delta(\mathrm{S}_{\deg f_{2}}(f_{1},f_{2})),N)$ is expected to equal one in practice, or a factorisation of $N$ is possibly obtained. However, as the requirement is new, its meaning is discussed before proceeding to the proof of the theorem.

Let $\mathbb{A}$ be an integral domain, $f_{1},f_{2}\in\mathbb{A}[x]$ such that $2\leq\deg f_{2}\leq\deg f_{1}$, and $d_{i}=\deg f_{i}$ for $i=1,2$. Then the \emph{first subresultant} of $f_{1}$ and $f_{2}$ is the polynomial
\begin{equation*}
	\sres_{1}(f_{1},f_{2})=(-1)^{d_{1}+d_{2}-1}\left(\mathrm{M}_{d_{2},d_{1}+d_{2}-2}(f_{1},f_{2})x-\mathrm{M}_{d_{2},d_{1}+d_{2}-1}(f_{1},f_{2})\right).
\end{equation*}
Recall that $\res(f_{1},f_{2})=0$ if and only if $\deg\gcd(f_{1},f_{2})\geq 1$, where the gcd is computed over the field of fractions of $\mathbb{A}$. This statement is refined by considering the first subresultant of the polynomials: if $\res(f_{1},f_{2})=0$, then $\deg\gcd(f_{1},f_{2})=1$ if and only if $\sres_{1}(f_{1},f_{2})\neq 0$ (see, for instance,~\cite[Section~7.7.1]{mishra1993} or \cite{elkahoui2003}).

Suppose now that $f_{1}$ and $f_{2}$ are coprime integer polynomials. Write $f_{1}=\sum^{d_{1}}_{i=0}a_{1,i}x^{i}$ and $f_{2}=\sum^{d_{1}}_{i=0}a_{2,i}x^{i}$ such that the coefficients $a_{i,j}$ are integers. Then using the Laplace expansion (see~\cite[Section~33]{aitken1956}) to compute the determinant of $\syl(f_{1},f_{2})$ with respect to rows $d_{2}$ and $d_{1}+d_{2}$ shows that
\begin{equation*}
	\res(f_{1},f_{2})=(-1)^{d_{2}-1}\sum^{d_{1}}_{i=1}\left(a_{1,i}a_{2,0}-a_{1,0}a_{2,i}\right)\mathrm{M}_{d_{2},d_{1}+d_{2}-i}(f_{1},f_{2}).
\end{equation*}
Thus, $\Delta(\mathrm{S}_{d_{2}}(f_{1},f_{2}))$ divides $\res(f_{1},f_{2})$ and $\sres_{1}(f_{1},f_{2})$. Therefore, if $p$ is a prime that divides $\Delta(\mathrm{S}_{d_{2}}(f_{1},f_{2}))$ and does not divide $\lc(f_{1})\lc(f_{2})$, then the reductions of $f_{1}$ and $f_{2}$ modulo $p$ have a common factor of degree greater than one. Conversely, if the reductions of $f_{1}$ and $f_{2}$ modulo a prime $p$ are nonzero and have a common divisor of degree $w\geq 2$, then Gomez et al.~\cite{gomez2009} showed that $p^{w}$ divides $\res(f_{1},f_{2})$. Modifying their proof shows that $p^{w-1}$ divides the minors $\mathrm{M}_{\deg f_{2},i}$, and thus $p^{w-1}$ divides $\Delta(\mathrm{S}_{d_{2}}(f_{1},f_{2}))$. Hence, if $p$ is a prime that does not divide $\lc(f_{1})\lc(f_{2})$, then $p$ divides $\Delta(\mathrm{S}_{d_{2}}(f_{1},f_{2}))$ if and only if the reductions of $f_{1}$ and $f_{2}$ modulo $p$ have a common factor of degree greater than one. If $f_{1}$ and $f_{2}$ are a pair of number field sieve polynomials, then $\gcd(\lc(f_{1})\lc(f_{2}),N)$ is expected to equal one in practice. Therefore, the requirement that $\gcd(\Delta(\mathrm{S}_{\deg f_{2}}(f_{1},f_{2})),N)=1$ denies the possibility that $f_{1}$ and $f_{2}$ have a factor of degree greater than one modulo some factor of $N$.

\subsection{Proof of Theorem~\ref{thm:polys-to-gp}}\label{sec:proof-polys-to-gp}

In this section, $f_{1}$ and $f_{2}$ are integer polynomials such that $2\leq\deg f_{2}\leq\deg f_{1}$. Furthermore, let $d=\deg f_{1}$,
\begin{equation*}
	\vec{c}_{t}=\vec{c}_{t}(f_{1},f_{2})=(c_{t,d+t-2},\ldots,c_{t,0})%
	\quad\text{and}\quad%
	C_{t}=\left(c_{t,d+t-i-j}\right)_{i=1,\ldots,t;j=1,\ldots,d}
\end{equation*}
for $t=\deg f_{2},\ldots,d$. The first lemma of this section proves Property~\eqref{isGP} and Property~\eqref{gcdCondition} of Theorem~\ref{thm:polys-to-gp}:

\begin{lemma}\label{lem:polys-to-gp} If $\gcd(\Delta(\mathrm{S}_{\deg f_{2}}(f_{1},f_{2})),N)=1$ and there exists an integer $r$ that is a root of $f_{1}$ and $f_{2}$ modulo $N$, then $\vec{c}_{t}$ is a nonzero geometric progression with ratio $r$ modulo $N$ and $\gcd(c_{t,0},N)=\gcd(\lc(f_{1})^{t-\deg f_{2}},N)$ for $t=\deg f_{2},\ldots,d$.
\end{lemma}
\begin{proof} Suppose that $\gcd(\Delta(\mathrm{S}_{\deg f_{2}}(f_{1},f_{2})),N)=1$ and $r\in\Z$ is a root of $f_{1}$ and $f_{2}$ modulo $N$. Let $d_{2}=\deg f_{2}$. The lemma is proved in two steps: first, it is shown that $\vec{c}_{d_{2}}$ is a geometric progression with ratio $r$ modulo $N$; and second, it is shown that if $d_{2}<d$ and $\vec{c}_{t-1}$ is a geometric progression with ratio $r$ modulo $N$ for some $t\in\{d_{2}+1,\ldots,d\}$, then so too is $\vec{c}_{t}$.

Let $U$ be a $(d+d_{2}-1)\times(d+d_{2}-1)$ unimodular matrix such that
$S_{d_{2}}U$ is in Hermite normal form (as defined by Cohen~\cite[Definition~2.4.2]{cohen1993}). Performing elementary row operations on the columns of $\mathrm{S}_{d_{2}}$ does not change $\Delta(\mathrm{S}_{d_{2}})$, which is nonzero since $\gcd(\Delta(\mathrm{S}_{d_{2}}),N)=1$. Thus,
\begin{equation*}
	\mathrm{S}_{d_{2}}U=\begin{pmatrix} \vec{0}_{d+d_{2}-2} & H\end{pmatrix}
\end{equation*}
for some $(d+d_{2}-2)\times(d+d_{2}-2)$ matrix $H$ such that $\det H=\pm\Delta(\mathrm{S}_{d_{2}})$.

The first column vector of $U$ is a basis for the integer kernel of $\mathrm{S}_{d_{2}}$ (see \cite[Proposition~2.4.9]{cohen1993}). Thus, \eqref{eqn:ct-in-kernel} implies that $\vec{c}^{T}_{d_{2}}$ is equal to $\pm\Delta(\vec{c}_{d_{2}})$ times the first column vector of $U$. The definition of $\vec{c}_{d_{2}}$ implies that $\Delta(\vec{c}_{d_{2}})=\Delta(\mathrm{S}_{d_{2}})$. Therefore, $\vec{c}_{d_{2}}$ is nonzero and $\vec{c}_{d_{2}}U^{-T}=\pm\left(\Delta(\mathrm{S}_{d_{2}}),0,\ldots,0\right)$.

Let $\vec{r}=(r^{d+d_{2}-2},r^{d+d_{2}-3},\ldots,1)$ and $\vec{r}U^{-T}=(r_{1},\ldots,r_{d+d_{2}-1})$. Then $\mathrm{S}_{d_{2}}\vec{r}^{T}\equiv\vec{0}_{d+d_{2}-2}\pmod{N}$ since $r$ is a root of $f_{1}$ and $f_{2}$ modulo $N$. Consequently,
\begin{equation*}
	H(r_{2},\ldots,r_{d+d_{2}-1})^{T}%
	\equiv\left(\mathrm{S}_{d_{2}}U\right)\left(\vec{r}U^{-T}\right)^{T}%
	\equiv\mathrm{S}_{d_{2}}\vec{r}^{T}%
	\equiv\vec{0}_{d+d_{2}-2}\pmod{N}
\end{equation*}
Therefore, $(r_{2},\ldots,r_{d+d_{2}-1})\equiv\vec{0}^{T}_{d+d_{2}-2}\pmod{N}$ since $H$ is invertible modulo $N$. It follows that $\gcd(r_{1},N)=1$ since $\Delta(\vec{r}U^{-T})=\Delta(\vec{r})=1$. Hence,
\begin{equation*}
	\vec{c}_{d_{2}}U^{-T}\equiv \pm\frac{\Delta(\mathrm{S}_{d_{2}})}{r_{1}}\vec{r}U^{-T}\pmod{N}.
\end{equation*}
Multiplying both sides of this equation on the right by $U^{T}$ shows that $\vec{c}_{d_{2}}$ is a geometric progression with ratio $r$ modulo $N$ and $\gcd(c_{d_{2},0},N)=1$.

Suppose that $d_{2}<d$ and $\vec{c}_{t-1}$ for some $t\in\{d_{2}+1,\ldots,d\}$ is a nonzero geometric progression with ratio $r$ modulo $N$ and $\gcd(c_{t-1,0},N)=\gcd(\lc(f_{1})^{t-1-d_{2}},N)$. Let $f_{1}=\sum^{d}_{i=0}a_{1,i}x^{i}$ such that $a_{1,d},\ldots,a_{1,0}\in\Z$. Then
\begin{equation}\label{eqn:St-recursion}
\mathrm{S}_{t}(f_{1},f_{2})=\left(
\begin{array}{c|cccccc}
 a_{1,d} & a_{1,d-1} & \ldots & a_{1,0} & 0 & \ldots & 0\\\hline
 0       & \multicolumn{6}{c}{}\\[-1ex]
 \vdots  & \multicolumn{6}{c}{\mathrm{S}_{t-1}(f_{1},f_{2})}\\
 0       & \multicolumn{6}{c}{}
\end{array}\right).
\end{equation}
Therefore, deleting the $i$th column of $\mathrm{S}_{t}(f_{1},f_{2})$ and computing the determinant of the resulting matrix along its first row shows that
\begin{equation}\label{eqn:Mti-recursion}
\mathrm{M}_{t,i}=
	\begin{cases}
		\sum^{d}_{i=1}a_{1,d-i}\mathrm{M}_{t-1,i} & \text{if $i=1$},\\
		-a_{1,d}\mathrm{M}_{t-1,i-1} & \text{if $i\in\{2,\ldots,d+t-1\}$}.
	\end{cases}
\end{equation}
It follows that $\vec{c}_{t}$ is nonzero since $a_{1,d}$ and $\vec{c}_{t-1}$ are nonzero. Furthermore, $c_{t,0}=-a_{1,d}c_{t-1,0}$ and thus $\gcd(c_{t,0},N)=\gcd(\lc(f_{1})^{t-d_{2}},N)$.

By assumption, $\vec{c}_{t-1}$ is a geometric progression with ratio $r$ modulo $N$. Thus, $\mathrm{M}_{t-1,i}r\equiv\mathrm{M}_{t-1,i-1}$ for $i=2,\ldots,d+t-1$. Therefore, \eqref{eqn:Mti-recursion} implies that
\begin{equation*}
	\mathrm{M}_{t,i}r\equiv -a_{1,d}\mathrm{M}_{t-1,i-1}r\equiv -a_{1,d}\mathrm{M}_{t-1,i-2}\equiv \mathrm{M}_{t,i-1}\pmod{N}
\end{equation*}
for $i=3,\ldots,d+t-1$. The first entry of $\mathrm{S}_{t-1}\vec{c}^{T}_{t-1}$ is equal to $\sum^{d}_{i=0}a_{1,d-i}\mathrm{M}_{t-1,i+1}$. Consequently, \eqref{eqn:ct-in-kernel} implies that $\sum^{d}_{i=0}a_{1,d-i}\mathrm{M}_{t-1,i+1}=0$. Thus, \eqref{eqn:Mti-recursion} implies that
\begin{equation*}
	\mathrm{M}_{t,2}r%
	\equiv-a_{1,d}\mathrm{M}_{t-1,1}r%
	\equiv\sum^{d}_{i=1}a_{1,d-i}\mathrm{M}_{t-1,i+1}r%
	\equiv\sum^{d}_{i=1}a_{1,d-i}\mathrm{M}_{t-1,i}%
	\equiv\mathrm{M}_{t,1}\pmod{N}.
\end{equation*}
Hence, $\vec{c}_{t}$ is a geometric progression with ratio $r$ modulo $N$.
\end{proof}

Define the \emph{volume} of a real matrix $A$, denoted $\vol A$, to be $\sqrt{\det AA^{T}}$. For a matrix $A$ over a commutative ring, define $\vol^{2} A=\det AA^{T}$. If $A$ is the $0\times n$ empty matrix, then $\vol A=1$ and $\vol^{2} A =1$. If $A$ is an $m\times n$ matrix such that $m\leq n$, then $\vol A$ is nonzero if and only if $A$ has full rank. The volume function is multiplicative in the following sense: if $A$ is an $m\times m$ matrix and $B$ is an $m\times n$ matrix, then $\vol AB=\vol A \cdot\vol B$ and $\vol^{2} AB=\vol^{2} A \cdot\vol^{2} B$.

The following lemma provides the upper bound on $\norm{\vec{c}_{t}}_{2,s^{-1}}$ in Theorem~\ref{thm:polys-to-gp}. In particular, the lemma shows that if $\left|\sin\theta_{s}(f_{1},f_{2})\right|\norm{f_{1}}_{2,s}\norm{f_{2}}_{2,s}=O(N^{1/d})$ for some $s>0$, i.e., the polynomials are close the attaining the lower bound from Lemma~\ref{lem:resultantBound}, then $\norm{\vec{c}_{d}}_{2,s^{-1}}=O(N^{1-1/d})$.

\begin{lemma}\label{lem:ct-upper-bnd} The inequality
\begin{equation}\label{eqn:ct-upper-bnd}
	\norm{\vec{c}_{t}}_{2,s^{-1}}%
	\leq\left(\left|\sin\theta_{s}(f_{1},f_{2})\right|\norm{f_{1}}_{2,s}\right)^{t-1}%
	\left(s^{(\deg f_{2}-t)/2}\norm{f_{2}}_{2,s}\right)^{d-1}
\end{equation}
holds for $t=\deg f_{2},\ldots,d$ and all $s>0$.
\end{lemma}
\begin{proof} If $\mathrm{S}_{\deg f_{2}}(f_{1},f_{2})$ does not have full rank, then the recurrence relation~\eqref{eqn:Mti-recursion} implies that $\vec{c}_{t}$ is the zero vector for $t=\deg f_{2},\ldots,d$. Thus, if $\mathrm{S}_{\deg f_{2}}(f_{1},f_{2})$ does not have full rank, then the lemma holds trivially. Therefore, assume that $\mathrm{S}_{\deg f_{2}}(f_{1},f_{2})$ has full rank. Then the recurrence relation~\eqref{eqn:St-recursion} implies that $\mathrm{S}_{t}(f_{1},f_{2})$ has full rank for $t=\deg f_{2},\ldots,d$.

Let $t\in\{\deg f_{2},\ldots,d\}$, $s$ be a positive real number and
\begin{equation*}
	S=s^{-\frac{d+t-2}{2}}\cdot\diag\left(s^{d+t-2},s^{d+t-3},\ldots,1\right).
\end{equation*}
Then the Binet--Cauchy formula (see \cite[Section~36]{aitken1956}) implies that
\begin{equation}\label{eqn:volStS}
  \vol^{2}\left(\mathrm{S}_{t}(f_{1},f_{2})S\right)%
  =\sum^{d+t-1}_{i=1}\left(\mathrm{M}_{t,i}\frac{\det S}{s^{(d+t-2i)/2}}\right)^{2}%
  =\sum^{d+t-1}_{i=1}\frac{\mathrm{M}^{2}_{t,i}}{s^{d+t-2i}}%
  =\norm{\vec{c}_{t}}^{2}_{2,s^{-1}}.
\end{equation}
To complete the proof, Fischer's inequality~\cite{fischer08} is used to derive an upper bound on $\vol^{2}(\mathrm{S}_{t}(f_{1},f_{2})S)$ in a manner similar to the proof of \cite[Lemma~2.2]{coxon2011}.

Define matrices $A_{1},\ldots,A_{t}$ as follows: $A_{1}$ is the $0\times(d+t-1)$ empty matrix if $t=d$; $A_{1}=(x^{d-2}f_{2},\ldots,x^{t-1}f_{2})_{d+t-2}$ if $t\neq d$; and $A_{i}=(x^{t-i}f_{1},x^{t-i}f_{2})_{d+t-2}$ for $i=2,\ldots,t$. For $i=1,\ldots,t$, let $B_{i}$ be the $(d-t+2(i-1))\times(d+t-1)$ matrix obtained by arranging the matrices $A_{1},\ldots,A_{i}$ consecutively beneath each other. Then $B_{t}$ is obtained from $\mathrm{S}_{t}$ by permuting its rows. Thus, $\vol^{2}(B_{t}S)=\vol^{2}(\mathrm{S}_{t}(f_{1},f_{2})S)$. Moreover, as $\mathrm{S}_{t}(f_{1},f_{2})$ has full rank and $s\neq 0$, $(A_{i}S)(A_{i}S)^{T}$ and $(B_{i}S)(B_{i}S)^{T}$ are positive definite Hermitian for $i=1,\ldots,t$. Therefore, Fischer's inequality implies that $\vol^{2}(B_{i}S)\leq\vol^{2}(A_{i}S)\cdot\vol^{2}(B_{i-1}S)$ for $i=2,\ldots,t$. Hence,
\begin{equation}\label{eqn:volAt-bnd}
	\vol^{2}(B_{t}S)\leq\vol^{2}(A_{1}S)\cdot\vol^{2}(A_{2}S)\cdots\vol^{2}(A_{t}S).
\end{equation}

For $i=2,\ldots,d$,
\begin{equation*}
	\norm{(x^{d-i}f_{2})_{d+t-2}S}_{2}=s^{\frac{d+\deg f_{2}-t}{2}-(i-1)}\norm{f_{2}}_{2,s}.
\end{equation*}
Thus, if $t\neq d$, then Hadamard's determinant theorem~\cite{hadamard1893} implies that
\begin{equation}\label{eqn:volAt1-bnd}
	\vol^{2}(A_{1}S)\leq\prod^{d-t+1}_{i=2}s^{d-t-\deg f_{2}-2(i-1)}\norm{f_{2}}^{2}_{2,s}=\left(s^{(\deg f_{2}-1)}\norm{f_{2}}^{2}_{2,s}\right)^{d-t}.
\end{equation}
This inequality also holds trivially if $t=d$. The angle between the row vectors of $A_{i}S$ is $\theta_{s}(f_{1},f_{2})$ and $\norm{(x^{t-i}f_{1})_{d+t-2}S}_{2}=s^{(t/2)-(i-1)}\norm{f_{1}}_{2,s}$ for $i=2,\ldots,t$. Therefore, by computing $(A_{i}S)^{T}(A_{i}S)$ or by viewing $\vol A_{i}S$ as the area of the parallelogram generated by the row vectors of $A_{i}S$, it follows that
\begin{equation*}
	\vol(A_{i}S)%
	=s^{t-\frac{d-\deg f_{2}}{2}-2(i-1)}\norm{f_{1}}_{2,s}\norm{f_{2}}_{2,s}\left|\sin\theta_{s}(f_{1},f_{2})\right|%
	\quad\text{for $i=2,\ldots,t$}.
\end{equation*}
Combining this equation with \eqref{eqn:volAt-bnd} and \eqref{eqn:volAt1-bnd} yields the inequality
\begin{equation*}
	\vol^{2}(B_{t}S)\leq\left(\left|\sin\theta_{s}(f_{1},f_{2})\right|\norm{f_{1}}_{2,s}\right)^{2(t-1)}\left(s^{(\deg f_{2}-t)/2}\norm{f_{2}}_{2,s}\right)^{2(d-1)}.
\end{equation*}
Combining this inequality with \eqref{eqn:volStS} (recalling that $\vol^{2}(B_{t}S)=\vol^{2}(\mathrm{S}_{t}S)$) and computing roots yields \eqref{eqn:ct-upper-bnd}. Hence, as $t$ and $s$ were chosen arbitrarily, it follows that \eqref{eqn:ct-upper-bnd} holds for $t=\deg f_{2},\ldots,d$ and all $s>0$.
\end{proof}

\begin{remark}\label{rmk:ct-upper-bnd} The last equality of \eqref{eqn:volStS} holds for $t=\deg f_{2},\ldots,d$ and all $s>0$. Thus, the inequalities
\begin{equation*}
  \norm{\vec{c}_{t}}^{2}_{2,s^{-1}}%
  \geq\Delta(\mathrm{S}_{t}(f_{1},f_{2}))^{2}\left(s^{d+t-2}+s^{d+t-4}+\ldots+s^{-(d+t-2)}\right)%
  \geq\Delta(\mathrm{S}_{t}(f_{1},f_{2}))^{2}
\end{equation*}
hold for $t=\deg f_{2},\ldots,d$ and all $s>0$.
\end{remark}

Property~\eqref{kernelCondition} of Theorem~\ref{thm:polys-to-gp} is now proved. Combining this property with Property~\eqref{isGP} of theorem shows that the coefficient vectors of each pair of nonlinear number field sieve polynomials appear in the kernel of a matrix of the form \eqref{eqn:partialC} for some nonzero length $2d-1$ geometric progression $[c_{2d-2},\ldots,c_{0}]$, where $d$ is the maximum degree of the polynomials.

\begin{lemma}\label{lem:kernelCondition} The vectors $(f_{1})^{T}_{d}$ and $(f_{2})^{T}_{d}$ belong to the kernel of $\partial C_{t}$ for $t=\deg f_{2},\ldots,d$.
\end{lemma}
\begin{proof} The $i$th entry of $\partial C_{t}\cdot(f_{1})^{T}_{d}$ is equal to the $i$th entry of $\mathrm{S}_{t}(f_{1},f_{2})\cdot\vec{c}^{T}_{t}$ for $i=1,\ldots,t-1$. Similarly, the $i$th entry of $\partial C_{t}\cdot(f_{2})^{T}_{d}$ is equal to the $(d+i-1)$th entry of $\mathrm{S}_{t}(f_{1},f_{2})\cdot\vec{c}^{T}_{t}$ for $i=1,\ldots,t-1$. Thus, \eqref{eqn:ct-in-kernel} implies that $(f_{1})^{T}_{d}$ and $(f_{2})^{T}_{d}$ are in the kernel of $\partial C_{t}$ for $t=\deg f_{2},\ldots,d$.
\end{proof}

Denote by $\adj A$ the adjoint (or adjugate) of a square matrix $A$. The following lemma, from which the remaining properties of Theorem~\ref{thm:polys-to-gp} are deduced, shows that the adjoint of the Bezout matrix has entries which, up to sign, are minors of the Sylvester matrix:
\begin{lemma}\label{lem:adjBez} Let $g_{1}=\sum^{d}_{i=0}u_{i}x^{i}$ and $g_{2}=\sum^{d}_{i=0}v_{i}x^{i}$ such that $u_{0},\ldots,u_{d}$ and $v_{0},\ldots,v_{d}$ are algebraically independent indeterminates over $\Z$. Then
\begin{equation}\label{eqn:adjBez}
  \adj\bez(g_{1},g_{2})=(-1)^{d(d-1)/2}\left(\mathrm{M}_{d,i+j-1}(g_{1},g_{2})\right)_{i=1,\ldots,d; j=1,\ldots,d}.
\end{equation}
\end{lemma}
\begin{proof} Let $g_{1}=\sum^{d}_{i=0}u_{i}x^{i}$ and $g_{2}=\sum^{d}_{i=0}v_{i}x^{i}$ such that $u_{0},\ldots,u_{d}$ and $v_{0},\ldots,v_{d}$ are algebraically independent indeterminates over $\Z$. Then the algebraic independence of the coefficients implies that $\res(g_{1},g_{2})\in\Z[u_{0},\ldots,u_{d},v_{0},\ldots,v_{d}]$ is nonzero. Therefore, equation~\eqref{eqn:det-bezout} implies that it is sufficient to show that
\begin{equation}\label{eqn:adjBez-sufficient}
    \bez(g_{1},g_{2})\cdot\left(\mathrm{M}_{d,i+j-1}(g_{1},g_{2})\right)_{\begin{subarray}{l}i=1,\ldots,d\\ j=1,\ldots,d\end{subarray}}=(-1)^{d}\,\res(g_{1},g_{2})\cdot \id_{d},
\end{equation}
where $\id_{n}$ denotes the $n\times n$ identity matrix for all integers $n\geq 1$. 

Following \eqref{eqn:bezout-pols}, define
\begin{equation}\label{eqn:bezout-pols-generic}
  p_{i+1}=\left(\sum^{i}_{k=0}v_{d-i+k}x^{k}\right)g_{1}-\left(\sum^{i}_{k=0}u_{d-i+k}x^{k}\right)g_{2}\quad\text{for $i=0,\ldots,d-1$}.
\end{equation}
Then $\bez(g_{1},g_{2})=(p_{1},\ldots,p_{d})$. Define
\begin{equation}\label{eqn:adjBez-Hii-def}
   \mathrm{H}_{i,j}=\left(x^{d-j}p_{i},x^{d-2}g_{1},\ldots,g_{1},x^{d-2}g_{2},\ldots,g_{2}\right)\quad\text{for $1\leq i,j\leq d$}.
\end{equation}
The matrices $\mathrm{H}_{i,j}$ are square of order $2d-1$ since $\deg g_{i}=d$ for $i=1,2$, and
\begin{equation*}
  p_{i}=(u_{d-i}v_{d}-v_{d-i}u_{d})x^{d-1}+\text{lower order terms}\quad\text{for $i=1,\ldots,d$}.
\end{equation*}
Expanding the determinant of each matrix $\mathrm{H}_{i,j}$ along its first row shows that
\begin{equation*}
   \bez(g_{1},g_{2})\cdot\left(\mathrm{M}_{d,i+j-1}(g_{1},g_{2})\right)_{\begin{subarray}{l}i=1,\ldots,d\\ j=1,\ldots,d\end{subarray}}=\left(\det\mathrm{H}_{i,j}\right)_{\begin{subarray}{l}i=1,\ldots,d\\ j=1,\ldots,d\end{subarray}}.
\end{equation*} 
It follows from \eqref{eqn:bezout-pols-generic} that
\begin{equation}\label{eqn:adjBez-pi-identity}
  x^{d-j}p_{i}=\left(\sum^{d+i-j-1}_{k=d-j}v_{k-i+j+1}x^{k}\right)g_{1}-\left(\sum^{d+i-j-1}_{k=d-j}u_{k-i+j+1}x^{k}\right)g_{2}
\end{equation}
for $1\leq i,j\leq d$. Therefore, for indices $i$ and $j$ such that $1\leq i<j\leq d$, the determinant $\det\mathrm{H}_{i,j}$ is zero since the polynomial $x^{d-j}p_{i}$ is a linear combination of the polynomials $x^{d-2}g_{1},\ldots,g_{1}$ and $x^{d-2}g_{2},\ldots,g_{2}$. Similarly, for $1\leq j<i\leq d$,
\begin{align*}
  x^{d-j}p_{i}%
  &=x^{i-j-1}\left(g_{2}-\sum^{d-i}_{k=0}v_{k}x^{k}\right)g_{1}-x^{i-j-1}\left(g_{1}-\sum^{d-i}_{k=0}u_{k}x^{k}\right)g_{2}\\
  &=-\left(\sum^{d-j-1}_{k=i-j-1}v_{k-i+j+1}x^{k}\right)g_{1}+\left(\sum^{d-j-1}_{k=i-j-1}u_{k-i+j+1}x^{k}\right)g_{2},
\end{align*}
Thus, $\det\mathrm{H}_{i,j}=0$ for $1\leq j<i\leq d$. Consequently,
\begin{equation}\label{eqn:adjBez-diag}
   \bez(g_{1},g_{2})\cdot\left(\mathrm{M}_{d,i+j-1}(g_{1},g_{2})\right)_{\begin{subarray}{l}i=1,\ldots,d\\ j=1,\ldots,d\end{subarray}}=\diag\left(\det\mathrm{H}_{1,1},\ldots,\det\mathrm{H}_{d,d}\right).
\end{equation} 

It is now shown that
\begin{equation}\label{eqn:adjBez-Hii}
  \det\mathrm{H}_{i,i}=(-1)^{d}\,\res(g_{1},g_{2})\quad\text{for $i=1,\ldots,d$}.
\end{equation}
The special case $i=d$ has been proved, up to sign, by Sederberg, Goldman and Du~\cite[Proposition~2.3]{sederberg97}. Their arguments are modified to obtain \eqref{eqn:adjBez-Hii}. In particular, the proof proceeds by computing the determinants of the following matrices two ways:
\begin{equation}
  \overline{\mathrm{H}}_{i}=\left(x^{d-i}p_{i},x^{d-2}g_{1},\ldots,g_{1},x^{d-1}g_{2},\ldots,g_{2}\right)\quad\text{for $i=1,\ldots,d$}.
\end{equation}

Substituting $i=j$ into \eqref{eqn:adjBez-pi-identity} shows that
\begin{equation*}
  x^{d-i}p_{i}=\left(\sum^{d-1}_{k=d-i}v_{k+1}x^{k}\right)g_{1}-\left(\sum^{d-1}_{k=d-i}u_{k+1}x^{k}\right)g_{2}\quad\text{for $i=1,\ldots,d$}.
\end{equation*}
Therefore, by performing elementary row operations on $\overline{\mathrm{H}}_{i}$, it follows that
\begin{equation}\label{eqn:adjBez-Hii-proof1}
  \det\overline{\mathrm{H}}_{i}=\det\left(v_{d}x^{d-1}g_{1},x^{d-2}g_{1},\ldots,g_{1},x^{d-1}g_{2},\ldots,g_{2}\right)=v_{d}\res(g_{1},g_{2})
\end{equation}
for $i=1,\ldots,d$. The first column vector of $\overline{\mathrm{H}}_{i}$ contains $v_{d}$ in the $(d+1)$th coordinate and zeros elsewhere. Furthermore, the submatrix of $\overline{\mathrm{H}}_{i}$ obtained by deleting its first column and $(d+1)$th row is equal to $\mathrm{H}_{i,i}$. Therefore, expanding the determinant of $\overline{\mathrm{H}}_{i}$ along its first column shows that
\begin{equation}\label{eqn:adjBez-Hii-proof2}
  \det\overline{\mathrm{H}}_{i}=(-1)^{d}v_{d}\det\mathrm{H}_{i,i}\quad\text{for $i=1,\ldots,d$}.
\end{equation}
Hence, combining \eqref{eqn:adjBez-Hii-proof1} and \eqref{eqn:adjBez-Hii-proof2} implies that \eqref{eqn:adjBez-Hii} holds. Then combining \eqref{eqn:adjBez-diag} and \eqref{eqn:adjBez-Hii} implies that \eqref{eqn:adjBez-sufficient} holds.
\end{proof}

Property~\eqref{rankCondition} of Theorem~\ref{thm:polys-to-gp} is now deduced from Lemma~\ref{lem:adjBez} by specialising the coefficients of the generic polynomials $g_{1}$ and $g_{2}$:

\begin{corollary}\label{cor:adjBez}  The adjoint of $\bez(f_{1},f_{2})$ is 
\begin{equation}\label{eqn:adjBez-spec}
  \adj\bez(f_{1},f_{2})=(-1)^{d(d-1)/2}C_{d}.
\end{equation}
Consequently, if $f_{1}$ and $f_{2}$ are coprime, then $C_{t}$ has full rank for $t=\deg f_{2},\ldots,d$.
\end{corollary}
\begin{proof} Let $\mathbb{A}=\Z[u_{0},\ldots,u_{d},v_{0},\ldots,v_{d}]$ such that $u_{0},\ldots,u_{d}$ and $v_{0},\ldots,v_{d}$ are algebraically independent indeterminates over $\Z$. Set $g_{1}=\sum^{d}_{i=0}u_{i}x^{i}\in\mathbb{A}[x]$ and $g_{2}=\sum^{d}_{i=0}v_{i}x^{i}\in\mathbb{A}[x]$. Then \eqref{eqn:adjBez} holds. Define the evaluation homomorphism $\varphi:\mathbb{A}\rightarrow\Z$ by $u_{i}\mapsto a_{1,i}$ and $v_{i}\mapsto a_{2,i}$ for $i=0,\ldots,d$. Extend $\varphi$ entry-wise to matrices. As $\deg f_{1}=d$, it holds that $\varphi(\adj\bez(g_{1},g_{2}))=\adj\bez(f_{1},f_{2})$ and $\varphi(\mathrm{M}_{d,t}(g_{1},g_{2}))=\mathrm{M}_{d,t}(f_{1},f_{2})$ for $i=1,\ldots,2d-1$. Therefore, the $\varphi$-image of the each side of \eqref{eqn:adjBez} is equal to its respective side of \eqref{eqn:adjBez-spec}.

Suppose now that $f_{1}$ and $f_{2}$ are coprime. Then \eqref{eqn:det-bezout} implies that $\bez(f_{1},f_{2})$ is nonsingular. Thus, \eqref{eqn:adjBez-spec} implies that $C_{d}$ is nonsingular. If $t\in\Z$ such that $\deg f_{2}\leq t\leq d$, then the recurrence relation~\eqref{eqn:St-recursion} implies that the submatrix of $C_{d}$ formed by its last $t$ rows is equal to $(-\lc(f_{1}))^{d-t}C_{t}$. As $\lc(f_{1})$ is nonzero, it follows that $C_{t}$ has full rank for $t=\deg f_{2},\ldots,d$.
\end{proof}

All that remains in the proof of Theorem~\ref{thm:polys-to-gp} is to establish the lower bound on $\norm{\vec{c}_{t}}_{2,s^{-1}}$ stated in Property~\eqref{GPsizeBounds} of the theorem. The remainder of this section is dedicated to the proof of this property, which proceeds as follows: first, the volume of $(\partial^{k}C_{t})S$, where $S$ is an arbitrary nonsingular matrix, is computed; then, for an appropriate choice of $S$, the volume of $(\partial^{k}C_{t})S$ is bounded above by a power of $\norm{\vec{c}_{t}}_{2,s^{-1}}$, providing a lower bound on $\norm{\vec{c}_{t}}_{2,s^{-1}}$; and finally, by examining a special case of this bound, the lower bound stated in Property~\eqref{GPsizeBounds} is proved.

Let $A$ be an $m\times n$ matrix. For all subsets $I\subseteq\{1,\ldots,m\}$ and $J\subseteq\{1,\ldots,n\}$, define $A_{I,J}$ to be the $\left|I\right|\times\left|J\right|$ submatrix of $A$ formed by the intersection of the rows that have indices in $I$ with the columns that have indices in $J$. If $m=n$, and $\{I,I'\}$ and $\{J,J'\}$ are partitions of $\{1,\ldots,n\}$ such that $\left|I\right|=\left|J\right|$, then Jacobi (see \cite[Section~42]{aitken1956} or \cite{brualdi1983}) showed that
\begin{equation}\label{eqn:jacobi-identity}
	\det\left(\adj A\right)_{I,J}%
	=(-1)^{\sum_{i'\in I'}i'+\sum_{j'\in J'}j'}%
	\left(\det A\right)^{|I|-1}%
	\det\left(A^{T}\right)_{I',J'}.
\end{equation}
The following technical lemma is proved by repeatedly applying this identity:

\begin{lemma}\label{lem:vol-adjoint} Suppose that $A$ and $S$ are $n\times n$ matrices such that $n\geq 2$ and $S$ is invertible. Then, for each partition $\{I,I'\}$ of $\{1,\ldots,n\}$,
\begin{equation*}\label{eqn:vol-adjoint}
	\vol^{2}\left(\left(\adj A\right)_{I,\{1,\ldots,n\}} S\right)%
	=\left(\det A\right)^{2(|I|-1)}\left(\det S\right)^{2}\vol^{2}\left(\left(A^{T}\right)_{I',\{1,\ldots,n\}} S^{-T}\right).
\end{equation*}
\end{lemma}
\begin{proof} Suppose that $A$ and $S$ are $n\times n$ matrices such that $n\geq 2$ and $S$ is invertible. Let $\{I,I'\}$ be a partition of $\{1,\ldots,n\}$. Set $B=\adj(S)A$, $\mathcal{J}=\{J\subseteq\{1,\ldots,n\}\mid |J|=|I|\}$ and $\mathcal{J}'=\{\{1,\ldots,n\}\setminus J\mid J\in\mathcal{J}\}$. Then the Binet--Cauchy formula implies that
\begin{equation*}
	\vol^{2}\left(\adj B\right)_{I,\{1,\ldots,n\}}%
	=\sum_{J\in\mathcal{J}}\left(\det\left(\adj B\right)_{I,J}\right)^{2}.
\end{equation*}
Using \eqref{eqn:jacobi-identity} to compute each summand on the right hand side shows that
\begin{equation*}
	\vol^{2}\left(\adj B\right)_{I,\{1,\ldots,n\}}%
	=\left(\det B\right)^{2(|I|-1)}\sum_{J'\in\mathcal{J}'}\left(\det \left(B^{T}\right)_{I',J'}\right)^{2}.
\end{equation*}
Using the Binet--Cauchy formula to compute the sum on the right hand side yields
\begin{equation}\label{eqn:vol-adjoint-Id}
	\vol^{2}\left(\adj B\right)_{I,\{1,\ldots,n\}}%
	=\left(\det B\right)^{2(|I|-1)}\vol^{2}\left(B^{T}\right)_{I',\{1,\ldots,n\}}.
\end{equation}
If $X$ and $Y$ are $n\times n$ matrices, then $(XY)_{K,\{1,\ldots,n\}}=X_{K,\{1,\ldots,n\}}Y$ for all $K\subseteq\{1,\ldots,n\}$. It follows that
\begin{equation*}
	\vol^{2}\left(\adj B\right)_{I,\{1,\ldots,n\}}=%
	\left(\det S\right)^{2(n-2)|I|}\vol^{2}\left(\left(\adj A\right)_{I,\{1,\ldots,n\}} S\right)
\end{equation*}
and
\begin{equation*}
	\vol^{2}\left(B^{T}\right)_{I',\{1,\ldots,n\}}=(\det S)^{2(n-|I|)}\vol^{2}\left(\left(A^{T}\right)_{I',\{1,\ldots,n\}} S^{-T}\right).
\end{equation*}
Substituting these values and $\det B=(\det S)^{n-1}\det A$ into \eqref{eqn:vol-adjoint-Id} completes the proof.
\end{proof}

\begin{lemma}\label{lem:gp-vol} Let $t,k\in\Z$ such that $\deg f_{2}\leq t\leq d$ and $0\leq k<t$, and $S$ be a real nonsingular $(d+k)\times(d+k)$ matrix. Then
\begin{multline}\label{eqn:gp-vol}
  \vol\left(\left(\partial^{k}C_{t}\right)S\right)%
  =\left|\det S\right|\left|\lc(f_{1})^{t-\deg f_{2}}\res(f_{1},f_{2})\right|^{t-k-1}\\
  \cdot\vol\big(\big(\underbrace{x^{k-1}f_{1},\ldots,f_{1}\vphantom{\big(}}_{\text{$k$ terms}},\underbrace{x^{d-t+k-1}f_{2},\ldots,f_{2}\vphantom{\big(}}_{\text{$d-t+k$ terms}}\big)_{d+k-1}S^{-T}\big).
\end{multline}
\end{lemma}
\begin{proof} Let $t,k\in\Z$ such that $\deg f_{2}\leq t\leq d$ and $0\leq k<t$, and $S$ be a real nonsingular $(d+k)\times(d+k)$ matrix. Let $\mathbb{A}=\R[u_{0},\ldots,u_{d},v_{0},\ldots,v_{d}]$ where $u_{0},\ldots,u_{d}$ and $v_{0},\ldots,v_{d}$ are algebraically independent indeterminates over $\R$. Define $g_{1}=\sum^{d}_{i=0}u_{i}x^{i}$, $g_{2}=\sum^{d}_{i=0}v_{i}x^{i}$ and $p_{1},\ldots,p_{d}\in\mathbb{A}[x]$ by \eqref{eqn:bezout-pols-generic}. Then $\bez(g_{1},g_{2})=(p_{1},\ldots,p_{d})$. Define a $k\times(d+k)$ matrix $G$ and a $(d+k)\times(d+k)$ matrix $B$ as follows:
\begin{equation*}
	G=\left(x^{k-1}g_{1},x^{k-2}g_{1},\ldots,g_{1}\right)_{d+k-1}%
	\quad\text{and}\quad%
	B=\left(\!\!\begin{array}{c|c}
	 	G^{T} & \!\!\begin{array}{c} 0_{k\times d}\\ \bez(g_{1},g_{2}) \end{array}
	\end{array}\!\!\!\!\right),
\end{equation*}
where $0_{m\times n}$ denotes the $m\times n$ matrix of zeros for all integers $m,n\geq 0$. The upper $k\times k$ submatrix of $G^{T}$ is lower triangular, with each entry on its diagonal equal to $u_{d}$. Thus, \eqref{eqn:det-bezout} implies that $\det B=u^{k}_{d}(-1)^{d(d+1)/2}\res(g_{1},g_{2})\in\mathbb{A}$, which is nonzero since $u_{0},\ldots,u_{d}$ and $v_{0},\ldots,v_{d}$ are algebraically independent over $\R$.

If $k\geq 1$, then \eqref{eqn:ct-in-kernel} implies that
\begin{align*}
	G\cdot\left(\mathrm{M}_{d,i+j-1}(g_{1},g_{2})\right)^{T}_{\begin{subarray}{l}j=1,\ldots,d+k\end{subarray}}%
	=\mathrm{S}_{d}(g_{1},g_{2})_{\{i,\ldots,i+k-1\},\{1,\ldots,2d-1\}}\cdot\vec{c}_{d}(g_{1},g_{2})^{T}=\vec{0}_{k}
\end{align*}
for $i=1,\ldots,d-k$. Consequently, Lemma~\ref{lem:adjBez} implies that
\begin{equation*}
	(-1)^{d(d-1)/2}u^{k}_{d}\left(\mathrm{M}_{d,i+j-1}(g_{1},g_{2})\right)_{\begin{subarray}{l}i=1,\ldots,d-k\\ j=1,\ldots,d+k\end{subarray}}\cdot B%
	=\begin{pmatrix}
		0_{(d-k)\times 2k} & \det B\cdot \id_{d-k}
	\end{pmatrix}.
\end{equation*}
As $B$ is nonsingular, it follows that
\begin{equation*}
	\left(\adj B\right)_{\{2k+1,\ldots,d+k\},\{1,\ldots,d+k\}}=(-1)^{d(d-1)/2}u^{k}_{d}\left(\mathrm{M}_{d,i+j-1}(g_{1},g_{2})\right)_{\begin{subarray}{l}i=1,\ldots,d-k\\ j=1,\ldots,d+k\end{subarray}}.
\end{equation*}
Therefore, on the one hand,
\begin{multline}\label{eqn:vol-adjB-i}
	\vol^{2}\left(\left(\adj B\right)_{\{d-t+2k+1,\ldots,d+k\},\{1,\ldots,d+k\}}S\right)\\
	=u^{2k(t-k)}_{d}\vol^{2}\left(\left(\mathrm{M}_{d,d-t+i+j-1}(g_{1},g_{2})\right)_{\begin{subarray}{l}i=1,\ldots,t-k\\ j=1,\ldots,d+k\end{subarray}}S\right).
\end{multline}
On the other hand, as $\bez(g_{1},g_{2})$ is symmetric (which is deduced from Lemma~\ref{lem:adjBez} by noting that $\adj\bez(g_{1},g_{2})$ is symmetric), Lemma~\ref{lem:vol-adjoint} implies that
\begin{multline}\label{eqn:vol-adjB-ii}
	\vol^{2}\left(\left(\adj B\right)_{\{d-t+2k+1,\ldots,d+k\},\{1,\ldots,d+k\}}S\right)\\
	=\left(u^{k}_{d}\res(g_{1},g_{2})\right)^{2(t-k-1)}(\det S)^{2}\\
 \cdot\vol^{2}\left(\left(x^{k-1}g_{1},\ldots,g_{1},p_{1},\ldots,p_{d-t+k}\right)_{d+k-1}S^{-T}\right).
\end{multline}

Write $f_{1}=\sum^{d}_{i=0}a_{1,i}x^{i}$ and $f_{2}=\sum^{d}_{i=0}a_{2,i}x^{i}$ such that the coefficients $a_{i,j}$ are integers. Define the evaluation homomorphism $\varphi:\mathbb{A}\rightarrow\R$ by $u_{i}\mapsto a_{1,i}$ and $v_{i}\mapsto a_{2,i}$ for $i=0,\ldots,d$. Then $\varphi(\res(g_{1},g_{2}))=a^{d-\deg f_{2}}_{1,d}\res(f_{1},f_{2})$. Extend $\varphi$ entry-wise to matrices and let $\tilde{\varphi}:\mathbb{A}[x]\rightarrow\R[x]$ be the natural extension of $\varphi$. Then
\begin{align*}
  \varphi\left(\left(\mathrm{M}_{d,d-t+i+j-1}(g_{1},g_{2})\right)_{\begin{subarray}{l}i=1,\ldots,t-k\\ j=1,\ldots,d+k\end{subarray}}\right)%
  &=\left(\mathrm{M}_{d,d-t+i+j-1}(f_{1},f_{2})\right)_{\begin{subarray}{l}i=1,\ldots,t-k\\ j=1,\ldots,d+k\end{subarray}}\\
  &=(-a_{1,d})^{d-t}\cdot\partial^{k}C_{t},
\end{align*}
where the final equality follows from the recurrence relation~\eqref{eqn:Mti-recursion}. Therefore, computing the $\varphi$-images of \eqref{eqn:vol-adjB-i} and \eqref{eqn:vol-adjB-ii} and equating shows that
\begin{multline}\label{eqn:gp-vol-ii}
  \vol^{2}\left(\left(\partial^{k}C_{t}\right)S\right)%
 =\left(\det S\right)^{2}\left(a^{t-\deg f_{2}}_{1,d}\res(f_{1},f_{2})\right)^{2(t-k-1)}a^{-2(d-t+k)}_{1,d}\\
 \cdot\vol^{2}\left(\left(x^{k-1}f_{1},\ldots,f_{1},\tilde{\varphi}(p_{1}),\ldots,\tilde{\varphi}(p_{d-t+k})\right)_{d+k-1}S^{-T}\right).
\end{multline}

From the definition of $p_{1},\ldots,p_{d}$, it follows that if $\deg f_{2}<d$, then
\begin{equation*}
  \tilde{\varphi}(p_{i})=-f_{2}\cdot\sum^{i-1}_{j=0}a_{1,d-i+1+j}x^{j}\quad\text{for $i=1,\ldots,d-\deg f_{2}$}.
\end{equation*}
Furthermore, if $d-t+k>d-\deg f_{2}$, then
\begin{equation*}
  \tilde{\varphi}(p_{d-\deg f_{2}+i})=\left(\sum^{i-1}_{j=0}a_{2,\deg f_{2}-i+1+j}x^{j}\right)f_{1}-\left(\sum^{d-\deg f_{2}+i-1}_{j=0}a_{1,\deg f_{2}-i+1+j}x^{j}\right)f_{2}
\end{equation*}
for $i=1,\ldots,\deg f_{2}-t+k$. As $\deg f_{2}-t+k-1\leq k-1$, it follows that
\begin{multline*}
	\vol^{2}\left(\left(x^{k-1}f_{1},\ldots,f_{1},\tilde{\varphi}(p_{1}),\ldots,\tilde{\varphi}(p_{d-t+k})\right)_{d+k-1}S^{-T}\right)\\
	=a^{2(d-t+k)}_{1,d}\cdot\vol^{2}\left(\left(x^{k-1}f_{1},\ldots,f_{1},x^{d-t+k-1}f_{2},\ldots,f_{2}\right)_{d+k-1}S^{-T}\right).
\end{multline*}
Substituting this equation into \eqref{eqn:gp-vol-ii} and taking roots yields \eqref{eqn:gp-vol}.
\end{proof}

\begin{lemma}\label{lem:ct-lower-bnd} Let $t,k\in\Z$ such that $\deg f_{2}\leq t\leq d$ and $0\leq k<t$. Then
\begin{multline}\label{eqn:ct-lower-bnd}
	\norm{\vec{c}_{t}}^{t-k}_{2,s^{-1}}%
	\geq s^{-\frac{t(d-t+k)+dk}{2}}\left|\lc(f_{1})^{t-\deg f_{2}}\res(f_{1},f_{2})\right|^{t-k-1}\\
	\cdot\vol\left(x^{k-1}f_{1}(sx),\ldots,f_{1}(sx),x^{d-t+k-1}f_{2}(sx),\ldots,f_{2}(sx)\right)_{d+k-1}
\end{multline}
for all $s>0$.
\end{lemma}
\begin{proof} Let $t,k\in\Z$ such that $\deg f_{2}\leq t\leq\deg f_{1}$ and $0\leq k<t$. For a real number $s>0$, define 
\begin{equation*}
	S_{1}=s^{-\frac{d+t-2}{2}}\diag\left(1,s,\ldots,s^{t-k-1}\right),%
	\quad%
	S_{2}=\diag\left(1,s,\ldots,s^{d+k-1}\right)
\end{equation*}
and
\begin{equation*}
	S_{3}=\diag\big(\underbrace{s^{-d},\ldots,s^{-(d+k-1)}}_{\text{$k$ terms}},s^{-t},\ldots,s^{-(d+k-1)}\big).
\end{equation*}
Then
\begin{multline*}
	\left(x^{k-1}f_{1},\ldots,f_{1},x^{d-t+k-1}f_{2},\ldots,f_{2}\right)_{d+k-1}S^{-T}_{2}\\
	=S_{3}\left(x^{k-1}f_{1}(sx),\ldots,f_{1}(sx),x^{d-t+k-1}f_{2}(sx),\ldots,f_{2}(sx)\right)_{d+k-1}.
\end{multline*}
Thus, Lemma~\ref{lem:gp-vol} with $S=S_{2}$ implies that
\begin{multline}\label{cor:vol-identity-gp-bnd-i}
  \vol\left(S_{1}\left(\partial^{k}C_{t}\right)S_{2}\right)\\
  =\left|\det S_{1}\right|\left|\det S_{2}\right|\left|\det S_{3}\right|\left|\lc(f_{1})^{t-\deg f_{2}}\res(f_{1},f_{2})\right|^{t-k-1}\\
  \cdot\vol\left(x^{k-1}f_{1}(sx),\ldots,f_{1}(sx),x^{d-t+k-1}f_{2}(sx),\ldots,f_{2}(sx)\right)_{d+k-1}.
\end{multline}

Recall that $\vec{c}_{t}=(c_{t,d+t-2},\ldots,c_{t,0})$ and $C_{t}=\left(c_{t,d+t-i-j}\right)_{i=1,\ldots,t;j=1,\ldots,d}$. Thus,
\begin{equation*}
	S_{1}\left(\partial^{k}C_{t}\right)S_{2}%
	=\left(c_{t,d+t-2-(i+j-2)}s^{(i+j-2)-\frac{d+t-2}{2}}\right)_{i=1,\ldots,t-1; j=1,\ldots,d+1}.
\end{equation*}
Therefore, the row vectors of $S_{1}\left(\partial^{k}C_{t}\right)S_{2}$ each have Euclidean length bounded by $\norm{\vec{c}_{t}}_{2,s^{-1}}$. Consequently, Hadamard's determinant theorem implies that
\begin{equation}\label{cor:vol-identity-gp-bnd-ii}
	\vol\left(S_{1}\left(\partial^{k}C_{t}\right)S_{2}\right)%
	\leq\norm{\vec{c}_{t}}^{t-1}_{2,s^{-1}}.
\end{equation}
Calculating the determinants of $S_{1}$, $S_{2}$ and $S_{3}$ yields
\begin{equation}\label{cor:vol-identity-gp-bnd-iii}
	\left|\det S_{1}\right|\left|\det S_{2}\right|\left|\det S_{3}\right|%
	=s^{-\frac{t(d-t+k)+dk}{2}}.
\end{equation}
Combining \eqref{cor:vol-identity-gp-bnd-i}, \eqref{cor:vol-identity-gp-bnd-ii} and \eqref{cor:vol-identity-gp-bnd-iii} gives \eqref{eqn:ct-lower-bnd}, which completes the proof since $s$ was chosen arbitrarily.
\end{proof}

To end the section, two corollaries to Lemma~\ref{lem:ct-lower-bnd} are given. The first corollary establishes the lower bound on $\norm{\vec{c}_{t}}_{2,s^{-1}}$ stated in Property~\eqref{GPsizeBounds} of Theorem~\ref{thm:polys-to-gp}, completing the proof of the theorem. The second corollary is utilised in the next section as part of the analysis of Montgomery's method.

\begin{corollary}\label{lem:ct-lower-bnd-k0} The inequality
\begin{equation*}
	\norm{\vec{c}_{t}}_{2,s^{-1}}%
	\geq s^{\frac{t-d}{2}}\left|s^{\deg f_{2}}\lc(f_{2})\right|^{\frac{d}{t}-1}\left|\lc(f_{1})^{t-\deg f_{2}}\res(f_{1},f_{2})\right|^{1-\frac{1}{t}}
\end{equation*}
holds for $t=\deg f_{2},\ldots,d$ and all $s>0$.
\end{corollary}
\begin{proof} For $t\in\Z$ such that $\deg f_{2}\leq t\leq d$ and $s>0$, the $(d-t)\times(d-t)$ submatrix of $(x^{d-t-1}f_{2}(sx),\ldots,f_{2}(sx))_{d-1}$ formed by columns $t-\deg f_{2}+1,\ldots,d-\deg f_{2}$ is upper triangular with $s^{\deg f_{2}}\lc(f_{2})$ in each diagonal entry. By applying the Binet--Cauchy formula, it follows that $\vol^{2}(x^{d-t-1}f_{2}(sx),\ldots,f_{2}(sx))_{d-1}$ is equal to $(s^{\deg f_{2}}\lc(f_{2}))^{2(d-t)}$ plus some sum of squares. Thus, for all $s>0$,
\begin{equation*}
	\vol\left(x^{d-t-1}f_{2}(sx),\ldots,f_{2}(sx)\right)_{d-1}%
	\geq\left|s^{\deg f_{2}}\lc(f_{2})\right|^{d-t}%
	\quad\text{for $t=\deg f_{2},\ldots,d$}.
\end{equation*}
Substituting these inequalities into \eqref{eqn:ct-lower-bnd} for $t=\deg f_{2},\ldots,d$ and $k=0$ completes the proof.
\end{proof}

\begin{corollary}\label{lem:ct-lower-bnd-td-k1} The inequality
\begin{equation*}
	\norm{\vec{c}_{d}}^{d-1}_{2,s^{-1}}%
	\geq s^{\frac{\deg f_{2}-d}{2}}\left|\lc(f_{1})^{d-\deg f_{2}}\res(f_{1},f_{2})\right|^{d-2}%
	\left|\sin\theta_{s}(f_{1},f_{2})\right|\norm{f_{1}}_{2,s}\norm{f_{2}}_{2,s}
\end{equation*}
holds for all $s>0$.
\end{corollary}
\begin{proof} For all $s>0$, the Euclidean length of the first row vector of the matrix $\left(f_{1}(sx),f_{2}(sx)\right)$ is $s^{d/2}\norm{f_{1}}_{2,s}$, the Euclidean length of the second row vector is $s^{\deg f_{2}/2}\norm{f_{2}}_{2,s}$, and the angle between the two row vectors is $\theta_{s}(f_{1},f_{2})$. Therefore,
\begin{equation*}
	\vol\left(f_{1}(sx),f_{2}(sx)\right)%
	=s^{\frac{d+\deg f_{2}}{2}}\norm{f_{1}}_{2,s}\norm{f_{2}}_{2,s}\left|\sin\theta_{s}(f_{1},f_{2})\right|%
	\quad\text{for all $s>0$}.
\end{equation*}
Substituting this equation into \eqref{eqn:ct-lower-bnd} for $t=d$ and $k=1$ completes the proof.
\end{proof}

\section{Analysis of Montgomery's method}\label{sec:analysis}

Montgomery's method is analysed in this section, providing criteria for the selection of geometric progressions that yield polynomials with optimal coefficient size and optimal resultant. In particular, the goal of this section is to prove the following theorem, which may be viewed as a converse to Theorem~\ref{thm:polys-to-gp}:
\begin{theorem}\label{thm:gp-to-polys} Let $d\geq 2$ and $\vec{c}=[c_{2d-2},\ldots,c_{0}]\in\Z^{2d-1}$ be a geometric progression with ratio $r$ modulo $N$ such that $C=C(c_{2d-2},\ldots,c_{0})$ is nonsingular and $\gcd(c_{0},\ldots,c_{d-2},N)=1$. Then, for $f_{1},f_{2}\in\Z[x]$ such that $2\leq\deg f_{2}\leq\deg f_{1}$ and $\left\{(f_{1})^{T}_{d},(f_{2})^{T}_{d}\right\}$ is a basis for the integer kernel of $\partial C$, the following properties hold:
\begin{enumerate}
	\item\label{degCondition} $\deg f_{1}=d$;
	\item\label{hasCommonRoot} $r$ is a root of $f_{1}$ and $f_{2}$ modulo $N$;
	\item\label{resFormula} $f_{1}$ and $f_{2}$ are coprime, with
	\begin{equation*}
		\left|\res(f_{1},f_{2})\right|=\frac{\left|\det C\right|^{d-1}}{\Delta(\partial C)^{\deg f_{2}}\Delta(\widehat{\partial C})^{d-\deg f_{2}}}%
		\quad\text{and}\quad%
		\Delta(\mathrm{S}_{d}(f_{1},f_{2}))=\Delta(\vec{c})\frac{\left|\det C\right|^{d-2}}{\Delta(\partial C)^{d-1}};
	\end{equation*}
	\item\label{polySizeBounds} the inequalities
	\begin{equation*}
  		\norm{\frac{\vec{c}}{\Delta(\vec{c})}}^{\frac{1}{d-1}}_{2,s^{-1}}%
  		\leq s^{\frac{\deg f_{2}-d}{2}}\left|\sin\theta_{s}(f_{1},f_{2})\right|\norm{f_{1}}_{2,s}\norm{f_{2}}_{2,s}%
  		\leq\frac{1}{N^{d-2}}\norm{\frac{\vec{c}}{\Delta(\vec{c})}}^{d-1}_{2,s^{-1}}
\end{equation*}
	hold for all $s>0$.
\end{enumerate}
\end{theorem}
Property~\eqref{degCondition} and Property~\eqref{hasCommonRoot} of Theorem~\ref{thm:gp-to-polys} are proved in Section~\ref{sec:montgomerys-method}. The two remaining properties of the theorem are proved in the next section.

Recall from Section~\ref{sec:montgomerys-method} that in Montgomery's method, the polynomials $f_{1}$ and $f_{2}$ are chosen such that $\{(f_{1}(sx))^{T}_{d},(f_{2}(sx))^{T}_{d}\}$, where $d=\max\{\deg f_{1},\deg f_{2}\}$, is a Lagrange-reduced basis for some $s>0$. It follows that $\left|\sin\theta_{s}(f_{1},f_{2})\right|\geq\sqrt{3}/2$ for the chosen value of $s$ (see \cite[p.~41]{nguyen2010}). Combining the inequalities from Property~\eqref{polySizeBounds} shows that any length $2d-1$ geometric progression $\vec{c}$ that satisfies the condition of the theorem has norm satisfying $\norm{\vec{c}/\Delta(\vec{c})}_{2,s^{-1}}\geq N^{1-1/d}$ for all $s>0$. Thus, Montgomery's method is unforgiving of a poor choice of geometric progression. In particular, the method generates to two degree $d$ polynomials with optimal coefficient size only if a geometric progression of almost minimal size is used.

Property~\eqref{resFormula} of Theorem~\ref{thm:gp-to-polys} may aid the selection of parameters for specific geometric progression constructions by allowing parameters to be tuned so that polynomials with resultant equal to a small multiple of $N$ are obtained. Before completing the proof of Theorem~\ref{thm:gp-to-polys} in the next section, Property~\eqref{resFormula} is used to compute the resultant given by two existing geometric progression constructions:

\begin{example}\label{ex:resFormula} For $d=2$, several authors~\cite{montgomery1993,williams2010,prest2012,koo2011,coxon2011} propose using length $2d-1=3$ geometric progressions of the form
\begin{equation*}
	[c_{2},c_{1},c_{0}]=\left[\frac{am^{2}-kN}{p},am,ap\right],
\end{equation*}
where $a$, $k$, $p$ and $m$ are nonzero integers such that $\gcd(m,p)=1$ and $\gcd(a,N)=1$. Letting $C=C(c_{2},c_{1},c_{0})$, it follows that
\begin{equation*}
	\det C=-akN%
	\quad\text{and}\quad%
	\Delta(\partial C)=\Delta\left([c_{2},c_{1},c_{0}]\right)=\gcd(a,c_{2}).
\end{equation*}
Let $\tilde{a}=a/\gcd(a,c_{2})$ and $\tilde{k}=k/\gcd(a,c_{2})$, with the latter being an integer since $am^{2}-kN=pc_{2}$ and $\gcd(a,N)=1$. Property~\eqref{resFormula} of Theorem~\ref{thm:gp-to-polys} implies that if $f_{1}$ and $f_{2}$ are quadratic polynomials whose coefficient vectors form a basis for the integer kernel of $\partial C$, then $\res(f_{1},f_{2})=\pm\tilde{a}\tilde{k}N$ and $\Delta(\mathrm{S}_{2}(f_{1},f_{2}))=1$.
\end{example}

\begin{example} For $d=3$, Koo, Jo and Kwon~\cite{koo2011} and the author~\cite{coxon2011} propose using length $2d-1=5$ geometric progressions of the form
\begin{equation*}
	[c_{4},c_{3},c_{2},c_{1},c_{0}]=\left[\frac{m(am^3 -kN)}{p^{2}},\frac{am^{3}-kN}{p},am^{2},amp,ap^{2}\right],
\end{equation*}
where $a$, $k$, $p$ and $m$ are nonzero integers such that $\gcd(m,p)=1$ and $\gcd(a,N)=1$. Letting $C=C(c_{4},\ldots,c_{0})$, it follows that $\det C=-a(kN)^{2}$,
\begin{equation*}
	\Delta(\partial C)=\gcd\left(\frac{am^{3}-kN}{p^{2}},am,ap\right)\cdot|k|N=\gcd\left(a,c_{3}/p\right)\cdot|k|N
\end{equation*}
and
\begin{equation*}
	\Delta([c_{4},\ldots,c_{0}])=\gcd\left(m\frac{am^3 -kN}{p^{2}},p\frac{am^{3}-kN}{p^{2}},a\right)=\gcd(a,c_{3}/p).
\end{equation*}
Let $\tilde{a}=a/\gcd(a,c_{3}/p)$ and $\tilde{k}=k/\gcd(a,c_{3}/p)$. Property~\eqref{resFormula} of Theorem~\ref{thm:gp-to-polys} implies that if $f_{1}$ and $f_{2}$ are cubic polynomials whose coefficient vectors form a basis for the integer kernel of $\partial C$, then $\res(f_{1},f_{2})=\pm\tilde{a}^{2}\tilde{k}N$ and $\Delta(\mathrm{S}_{3}(f_{1},f_{2}))=\left|\tilde{a}\right|$.
\end{example}

\subsection{Proof of Theorem~\ref{thm:gp-to-polys}}

Theorem~\ref{thm:polys-to-gp} shows that each pair of nonlinear number field sieve polynomials with maximum degree $d\geq 2$ appears in the kernel of the matrix $\partial C(c_{2d-2},\ldots,c_{0})$ for some geometric progression $[c_{2d-2},\ldots,c_{0}]$. The following lemma shows that such a geometric progression is unique up to scalar multiple, thus allowing results from Section~\ref{sec:proof-polys-to-gp} to be used in the proof of Theorem~\ref{thm:gp-to-polys}:

\begin{lemma}\label{lem:gp-kernel} Let $f_{1}$ and $f_{2}$ be coprime integer polynomials such that $2\leq\deg f_{2}\leq\deg f_{1}$, and $d=\deg f_{1}$. Then the vectors $(f_{1})^{T}_{d}$ and $(f_{2})^{T}_{d}$ are in the kernel of $\partial C(c_{2d-2},\ldots,c_{0})$ for some vector $\vec{c}=(c_{2d-2},\ldots,c_{0})\in\Z^{2d-1}$ if and only if $\vec{c}=\pm(\Delta(\vec{c})/\Delta(\vec{c}_{d}(f_{1},f_{2}))\cdot\vec{c}_{d}(f_{1},f_{2})$.
\end{lemma}
\begin{proof} Let $f_{1}$ and $f_{2}$ be coprime integer polynomials such that $2\leq\deg f_{2}\leq\deg f_{1}$, $d=\deg f_{1}$ and $\vec{c}_{d}=\vec{c}_{d}(f_{1},f_{2})$. Let $\vec{c}=(c_{2d-2},\ldots,c_{0})\in\Z^{2d-1}$ and $C=C(c_{2d-2},\ldots,c_{0})$. Then the $i$th entry of $\partial C\cdot(f_{1})^{T}_{d}$ is equal to the $i$th entry of $\mathrm{S}_{d}(f_{1},f_{2})\cdot\vec{c}^{T}$ for $i=1,\ldots,d-1$. Similarly, the $i$th entry of $\partial C\cdot(f_{2})^{T}_{d}$ is equal to the $(d+i-1)$th entry of $\mathrm{S}_{d}(f_{1},f_{2})\cdot\vec{c}^{T}$ for $i=1,\ldots,d-1$. Thus, $(f_{1})^{T}_{d}$ and $(f_{2})^{T}_{d}$ are in the kernel of $\partial C$ if and only if $\vec{c}^{T}$ is in the kernel of $\mathrm{S}_{d}(f_{1},f_{2})$.

As $f_{1}$ and $f_{2}$ are coprime, $\syl(f_{1},f_{2})$ is nonsingular. The matrix $\mathrm{S}_{\deg f_{2}}(f_{1},f_{2})$ is the submatrix of $\syl(f_{1},f_{2})$ obtained by first deleting its first and $(\deg f_{2}+1)$th rows, giving a matrix whose first column contains zeros, then deleting the first column of the resulting matrix. Therefore, $\mathrm{S}_{\deg f_{2}}(f_{1},f_{2})$ has full rank. Consequently, the recurrence relation~\eqref{eqn:St-recursion} implies that $\mathrm{S}_{d}(f_{1},f_{2})$ has full rank. Lemma~\ref{lem:polys-to-gp} and \eqref{eqn:ct-in-kernel} imply that $\vec{c}^{T}_{d}$ is nonzero and belongs to the kernel of $\mathrm{S}_{d}(f_{1},f_{2})$. Thus, $\{\vec{c}^{T}_{d}/\Delta(\vec{c}_{d})\}$ is a basis of the integer kernel of $\mathrm{S}_{d}(f_{1},f_{2})$. Hence, $\vec{c}^{T}$ is in the integer kernel of $\mathrm{S}_{d}(f_{1},f_{2})$ if and only if $\vec{c}=\pm(\Delta(\vec{c})/\Delta(\vec{c}_{d}))\cdot\vec{c}_{d}$.
\end{proof}

Let $n\geq 2$ and $A$ be an $n\times n$ matrix. For $i_{1},i_{2},j_{1},j_{2}\in\{1,\ldots,n\}$ such that $i_{1}<i_{2}$ and $j_{1}<j_{2}$, define $I_{k}=\{1,\ldots,n\}\setminus\{i_{k}\}$ and $J_{k}=\{1,\ldots,n\}\setminus\{j_{k}\}$ for $k=1,2$. Then
\begin{equation}\label{eqn:sylvester-identity}
	\det A\cdot\det A_{I_{1}\cap I_{2},J_{1}\cap J_{2}}%
	=\det\begin{pmatrix}
		\det A_{I_{1},J_{1}} & \det A_{I_{1},J_{2}}\\
		\det A_{I_{2},J_{1}} & \det A_{I_{2},J_{2}}
	\end{pmatrix},
\end{equation}
which is known as Sylvester's identity (see \cite[Theorem~4.1]{tyrtyshnikov2010} for a proof).

The following lemma is obtained from Lemma~\ref{lem:HankelInverse} by computing the determinant of the matrix $\psi$ for the case where $H$ and the polynomials $f_{1}$ and $f_{2}$ are integral:

\begin{lemma}\label{lem:adjC} Let $d\geq 2$ and $(c_{2d-2},\ldots,c_{0})\in\Z^{2d-1}$ such that $C=C(c_{2d-2},\ldots,c_{0})$ is nonsingular. If $f_{1}$ and $f_{2}$ are integer polynomials such that $\left\{(f_{1})^{T}_{d},(f_{2})^{T}_{d}\right\}$ is a basis of the integer kernel of $\partial C$, then
\begin{equation}\label{eqn:adjC}
	\adj C=\pm\Delta(\partial C)\cdot\bez(f_{1},f_{2}).
\end{equation}
\end{lemma}
\begin{proof} Let $d\geq 2$ and $(c_{2d-2},\ldots,c_{0})\in\Z^{2d-1}$ such that $C=C(c_{2d-2},\ldots,c_{0})$ is nonsingular. Define $(d+1)\times(d+1)$ matrices
\begin{equation*}
	B_{i,j}=%
	\begin{pmatrix}
	& & & 1 & & & & \\
	& & & & & 1 & & \\
	c_{2d-2} & c_{2d-3} & \ldots & c_{2d-i-1} & \ldots & c_{2d-j-1} & \ldots & c_{d-2} \\
	c_{2d-3} & c_{2d-4} & \ldots & c_{2d-i-2} & \ldots & c_{2d-j-2} & \ldots & c_{d-3} \\
	\vdots & \vdots &    & \vdots & & \vdots & & \vdots \\
	c_{d} & c_{d-1} & \ldots & c_{d-i+1} & \ldots & c_{d-j+1} & \ldots & c_{0}
	\end{pmatrix},
\end{equation*}
where the omitted entries are zeros, for $1\leq i,j\leq d+1$. Define $b_{i,j}=\det B_{i,j}$ for $1\leq i,j\leq d+1$. Then $\{(-1)^{i+j+1}b_{i,j}\mid 1\leq i<j\leq d\}$ is the set of $(d-1)\times(d-1)$ minors of $\partial C$. The matrix $\partial C$ has full rank since $C$ is nonsingular. Therefore, there exist indices $k$ and $\ell$ such that $k<\ell$ and $b_{k,\ell}\neq 0$. The first and second column vectors of $\adj B_{k,\ell}$ are $\left(b_{1,\ell},\ldots,b_{d+1,\ell}\right)^{T}$ and $\left(b_{k,1},\ldots,b_{k,d+1}\right)^{T}$ respectively. Thus, these two vectors are nonzero and in the kernel of $\partial C$.

Let
\begin{equation*}
	B=\begin{pmatrix}
		b_{k,1} & \ldots & b_{k,d+1}\\
		b_{1,\ell} & \ldots & b_{d+1,\ell}\\
	\end{pmatrix}.
\end{equation*}
Let $j_{1}$ and $j_{2}$ be indices such that $1\leq j_{1}<j_{2}\leq d$. Define $I_{n}=\{1,\ldots,d+1\}\setminus\{n\}$ and $J_{n}=\{1,\ldots,d+1\}\setminus\{j_{n}\}$ for $n=1,2$. Then
\begin{equation*}
	\det (B_{k,\ell})_{I_{1},J_{n}}%
	=\begin{cases}
		(-1)^{j_{n}+1}b_{j_{n},\ell} & \text{if $j_{n}\leq\ell$}\\
		(-1)^{j_{n}}b_{\ell,j_{n}} & \text{if $j_{n}>\ell$}
	\end{cases}%
	=(-1)^{j_{n}+1}b_{j_{n},\ell}%
	\quad\text{for $n=1,2$.}
\end{equation*}
Similarly, $\det (B_{k,\ell})_{I_{2},J_{n}}=(-1)^{j_{n}}b_{k,j_{n}}$ for $n=1,2$. Finally,
\begin{equation*}
	\det (B_{k,\ell})_{I_{1}\cap I_{2},J_{1}\cap J_{2}}=(-1)^{j_{1}+j_{2}+1}b_{j_{1},j_{2}}.
\end{equation*}
Therefore, identity~\eqref{eqn:sylvester-identity} implies that
\begin{equation}\label{eqn:sylvester-identity-Cij}
	\det\begin{pmatrix} b_{k,j_{1}} & b_{k,j_{2}}\\ b_{j_{1},\ell} & b_{j_{2},\ell}\end{pmatrix}=-b_{k,\ell}b_{j_{1},j_{2}}\quad\text{for $1\leq j_{1}<j_{2}\leq d+1$}.
\end{equation}
Consequently, $\Delta(B)=\pm\Delta(\partial C)\cdot b_{k,\ell}$, which is nonzero since $\partial C$ has full rank and $b_{k,\ell}\neq 0$. Therefore, if $V$ is a $(d+1)\times(d+1)$ unimodular matrix such that $BV$ is in Hermite normal form, then
\begin{equation*}
	BV=\begin{pmatrix}
		0_{2\times(d-1)} & H
	\end{pmatrix}
\end{equation*}
for some $2\times 2$ integer matrix $H$. Performing elementary column operations on $B$ does not change $\Delta(B)$. Thus, $\det H=\pm\Delta(B)$ is nonzero. Hence, 
\begin{equation}\label{eqn:kerPartialCBasis}
	H^{-1}B%
	=\begin{pmatrix}
		0_{2\times(d-1)} & \id_{2}
	\end{pmatrix}V^{-1}
\end{equation}
has integer entries and $(\partial C)(H^{-1}B)^{T}=0_{2\times 2}$.

Suppose that $f_{1}$ and $f_{2}$ are integer polynomials such that $\left\{(f_{1})^{T}_{d},(f_{2})^{T}_{d}\right\}$ is a basis of the integer kernel of $\partial C$. Then there exists a $2\times 2$ nonsingular integer matrix $U$ such that
\begin{equation}\label{eqn:fbasis-to-Bbasis}
	U\left(f_{1},f_{2}\right)=H^{-1}B.
\end{equation}
Therefore, Lemma~\ref{lem:HankelInverse} implies that
\begin{equation}\label{eqn:Cinverse}
  C^{-1} = -\frac{1}{\det\psi}\bez(f_{1},f_{2}),
\end{equation}
where
\begin{equation*}
	\psi=\begin{pmatrix}
  		c_{d-1} & \ldots & c_{0} & 0\\
  		0 & \ldots & 0 & 1
	\end{pmatrix}%
	\left(f_{1},f_{2}\right)^{T}
	=\begin{pmatrix}
  		c_{d-1} & \ldots & c_{0} & 0\\
  		0 & \ldots & 0 & 1
  	\end{pmatrix}B^{T}(HU)^{-T}.
\end{equation*}
Equation~\eqref{eqn:kerPartialCBasis} implies that $\Delta(H^{-1}B)=1$ since $V$ is unimodular. Furthermore, $\Delta((f_{1},f_{2}))=1$ since $\left\{(f_{1})^{T}_{d},(f_{2})^{T}_{d}\right\}$ is a basis of the integer kernel of $\partial C$: if $\Delta((f_{1},f_{2}))$ were greater than one, then $(f_{1})^{T}_{d}$ and $(f_{2})^{T}_{d}$ would only generate a proper subgroup of the integer kernel \cite[Section~3.1]{coxon2011}. Thus, \eqref{eqn:fbasis-to-Bbasis} implies that $U$ is unimodular. Therefore,
\begin{equation}\label{eqn:detPsi}
	\det\psi=\pm\frac{1}{\det H}\sum^{d}_{j=1}c_{d-j}\det\begin{pmatrix} b_{k,j} & b_{k,d+1}\\ b_{j,\ell} & b_{d+1,\ell}\end{pmatrix}=\pm\frac{1}{\Delta(\partial C)}\sum^{d}_{j=1}c_{d-j}b_{j,d+1}.
\end{equation}
Expanding the determinant of $C$ by minors along its last row shows that
\begin{equation}\label{eqn:detCminors}
	\sum^{d}_{i=1}c_{d-i}b_{i,d+1}=\det C.
\end{equation}
Hence, \eqref{eqn:Cinverse}, \eqref{eqn:detPsi} and \eqref{eqn:detCminors} imply that $\adj C=\pm\Delta(\partial C)\cdot\bez(f_{1},f_{2})$.
\end{proof}

The first assertion of Property~\eqref{resFormula} of Theorem~\ref{thm:gp-to-polys}, that $f_{1}$ and $f_{2}$ are coprime, follows from Lemma~\ref{eqn:adjC} since $C$ is nonsingular by assumption. The next step in the proof of the theorem is to prove the formulae for $\res(f_{1},f_{2})$ and $\Delta(\mathrm{S}_{d}(f_{1},f_{2}))$.
\begin{lemma}\label{lem:resFormula} Let $d\geq 2$ and $\vec{c}=(c_{2d-2},\ldots,c_{0})\in\Z^{2d-1}$ such that the matrix $C=C(c_{2d-2},\ldots,c_{0})$ is nonsingular. If $f_{1},f_{2}\in\Z[x]$ such that $2\leq\deg f_{2}\leq\deg f_{1}$ and $\left\{(f_{1})^{T}_{d},(f_{2})^{T}_{d}\right\}$ is a basis for the integer kernel of $\partial C$, then
\begin{equation*}
	\lc(f_{1})^{d-\deg f_{2}}\res(f_{1},f_{2})=\pm\frac{\left(\det C\right)^{d-1}}{\Delta(\partial C)^{d}}%
	\quad\text{and}\quad%
	\Delta\left(\mathrm{S}_{d}(f_{1},f_{2})\right)=\Delta(\vec{c})\frac{\left|\det C\right|^{d-2}}{\Delta(\partial C)^{d-1}}.
\end{equation*}
\end{lemma}
\begin{proof} Let $d\geq 2$ and $\vec{c}=(c_{2d-2},\ldots,c_{0})\in\Z^{2d-1}$ such that $C=C(c_{2d-2},\ldots,c_{0})$ is nonsingular. Suppose that $f_{1},f_{2}\in\Z[x]$ such that $2\leq\deg f_{2}\leq\deg f_{1}$ and $\left\{(f_{1})^{T}_{d},(f_{2})^{T}_{d}\right\}$ is a basis for the integer kernel of $\partial C$. Then \eqref{eqn:adjC} holds and computing the determinant of both sides of the equation yields
\begin{equation*}
	(\det C)^{d-1}=\pm\Delta(\partial C)^{d}\lc(f_{1})^{d-\deg f_{2}}\res(f_{1},f_{2}).
\end{equation*}
As $C$ is nonsingular, $\det C$ and $\Delta(\partial C)$ are nonzero. Thus, $\res(f_{1},f_{2})$ is nonzero and Lemma~\ref{lem:gp-kernel} implies that $\vec{c}=\pm\left(\Delta(\vec{c})/\Delta(\vec{c}_{d}(f_{1},f_{2}))\right)\cdot\vec{c}_{d}(f_{1},f_{2})$. Therefore, on the one hand, Corollary~\ref{cor:adjBez} implies that
\begin{equation*}
	\adj\bez(f_{1},f_{2})=\pm\frac{\Delta(\vec{c}_{d}(f_{1},f_{2}))}{\Delta(\vec{c})}C=\pm\frac{\Delta\left(\mathrm{S}_{d}(f_{1},f_{2})\right)}{\Delta(\vec{c})}C
\end{equation*}
On the other hand, computing the adjoint of both side of \eqref{eqn:adjC} yields
\begin{equation*}
	(\det C)^{d-2}C=\pm\Delta(\partial C)^{d-1}\adj\bez(f_{1},f_{2}).
\end{equation*}
As $C$ has at least one nonzero entry, it follows that
\begin{equation*}
	\pm(\det C)^{d-2}=\Delta(\partial C)^{d-1}\frac{\Delta\left(\mathrm{S}_{d}(f_{1},f_{2})\right)}{\Delta(\vec{c})},
\end{equation*}
where the right hand side is positive.
\end{proof}

The following lemma and its subsequent corollary complete the proof of Property~\eqref{resFormula} of Theorem~\ref{thm:gp-to-polys} by showing that $\lc(f_{1})=\pm\Delta(\widehat{\partial C})/\Delta(\partial C)$ when the degree of $f_{2}$ is strictly less than $d$:

\begin{lemma}\label{lem:basis_kerParC} Let $d\geq 2$, $(c_{2d-2},\ldots,c_{0})\in\Z^{2d-1}$ such that $C=C(c_{2d-2},\ldots,c_{0})$ is nonsingular, and integers $b_{i,j}$, for $1\leq i,j\leq d+1$, be defined as in the proof of Lemma~\ref{lem:adjC}. Then there exist integers $x_{2},\ldots,x_{d+1}$ such that $\sum^{d+1}_{k=2}x_{k}b_{k,1}=\Delta(\widehat{\partial C})$ and, for any such integers, the set
\begin{equation*}
	\left\{\frac{1}{\Delta(\partial C)}\sum^{d+1}_{k=2}x_{k}\left(b_{k,1},b_{k,2},\ldots,b_{k,d+1}\right)^{T},%
	\frac{1}{\Delta(\widehat{\partial C})}\left(b_{1,1},b_{2,1},\ldots,b_{d+1,1}\right)^{T}\right\}
\end{equation*}
is a basis of the integer kernel of $\partial C$.
\end{lemma}
\begin{proof} Let $(c_{2d-2},\ldots,c_{0})\in\Z^{2d-1}$ such that $C=C(c_{2d-2},\ldots,c_{0})$ is nonsingular Then $\Delta(\widehat{\partial C})$ is nonzero since $C$ is nonsingular. Define $b_{i,j}$ as in the proof of Lemma~\ref{lem:adjC} for $1\leq i,j\leq d+1$. For distinct indices $i$ and $j$, $b_{i,j}$ is up to sign equal to the determinant of the $(d-1)\times(d-1)$ submatrix of $\partial C$ obtained by deleting columns $i$ and $j$. Thus, $\gcd(b_{2,1},b_{3,1},\ldots,b_{d+1,1})=\Delta(\widehat{\partial C})$ and there exist integers $x_{2},\ldots,x_{d+1}$ such that $\sum^{d+1}_{k=2}x_{k}b_{k,1}=\Delta(\widehat{\partial C})$. For such integers, define
\begin{equation*}
	\vec{b}_{1}=\frac{1}{\Delta(\partial C)}\sum^{d+1}_{k=2}x_{k}\left(b_{k,1},\ldots,b_{k,d+1}\right)^{T}%
	\quad\text{and}\quad%
	\vec{b}_{2}=\frac{1}{\Delta(\widehat{\partial C})}\left(b_{1,1},\ldots,b_{d+1,1}\right)^{T}.
\end{equation*}
Write $\vec{b}_{i}=(\beta_{i,1},\ldots,\beta_{i,d+1})^{T}$ for $i=1,2$, and let $B=(\beta_{i,j})_{i=1,2;j=1,\ldots,d+1}$. As $b_{i,i}=0$ for $i=1,\ldots,d+1$, it follows that $\vec{b}_{1}$ and $\vec{b}_{2}$ have integer entries. Moreover, it follows from the proof of Lemma~\ref{lem:adjC} that $\vec{b}_{1}$ and $\vec{b}_{2}$ belong to the integer kernel of $\partial C$. Therefore, $\{\vec{b}_{1},\vec{b}_{2}\}$ is a basis for the integer kernel of $\partial C$ if and only if $\Delta(B)=1$: $\vec{b}_{1}$ and $\vec{b}_{2}$ are linearly independent if and only if $\Delta(B)\neq 0$; and if $\Delta(B)\neq 0$, then $\{\vec{b}_{1},\vec{b}_{2}\}$ is a basis of an index $\Delta(B)$ subgroup of the integer kernel~\cite[Section~3.1]{coxon2011}.

For $1\leq i<j\leq d+1$,
\begin{equation*}
	\det\begin{pmatrix}
		\beta_{1,i} & \beta_{1,j}\\ \beta_{2,i} & \beta_{2,j}
	\end{pmatrix}
	=\frac{1}{\Delta(\partial C)\Delta(\widehat{\partial C})}\sum^{d+1}_{k=2}x_{k}\det\begin{pmatrix}
		b_{k,i} & b_{k,j}\\ b_{i,1} & b_{j,1}
	\end{pmatrix}.
\end{equation*}
Therefore, \eqref{eqn:sylvester-identity-Cij} implies that
\begin{equation*}
	\det\begin{pmatrix}
		\beta_{1,i} & \beta_{1,j}\\ \beta_{2,i} & \beta_{2,j}
	\end{pmatrix}
	=\frac{1}{\Delta(\partial C)\Delta(\widehat{\partial C})}\sum^{d+1}_{k=2}-x_{k}b_{k,1}b_{i,j}=-\frac{b_{i,j}}{\Delta(\partial C)}
\end{equation*}
for $1\leq i<j\leq d+1$. The greatest common divisor of the $b_{i,j}$, for $1\leq i<j\leq d+1$, is equal to $\Delta(\partial C)$. Hence, $\Delta(B)=1$.
\end{proof}

\begin{corollary}\label{cor:basis-kerParC} Let $d\geq 2$ and $(c_{2d-2},\ldots,c_{0})\in\Z^{2d-1}$ such that the matrix $C=C(c_{2d-2},\ldots,c_{0})$ is nonsingular.  If $f_{1},f_{2}\in\Z[x]$ such that $\deg f_{2}\leq\deg f_{1}$ and $\left\{(f_{1})^{T}_{d},(f_{2})^{T}_{d}\right\}$ is a basis for the integer kernel of $\partial C$, then $\Delta(\widehat{\partial C})/\Delta(\partial C)$ divides $\lc(f_{1})$. Furthermore, if $\deg f_{2}<\deg f_{1}$, then $\lc(f_{1})=\pm\Delta(\widehat{\partial C})/\Delta(\partial C)$.
\end{corollary}
\begin{proof} Let $d\geq 2$ and $(c_{2d-2},\ldots,c_{0})\in\Z^{2d-1}$ such that $C=C(c_{2d-2},\ldots,c_{0})$ is nonsingular. Define $\vec{b}_{1}=(\beta_{1,1},\ldots,\beta_{1,d+1})^{T}$ and $\vec{b}_{2}=(\beta_{2,1},\ldots,\beta_{2,d+1})^{T}$ as in the proof of Lemma~\ref{lem:basis_kerParC}. Then $\{\vec{b}_{1},\vec{b}_{2}\}$ is a basis for the integer kernel of $\partial C$,
\begin{equation*}
	\beta_{1,1}=\frac{1}{\Delta(\partial C)}\sum^{d+1}_{k=2}x_{k}b_{k,1}=\frac{\Delta(\widehat{\partial C})}{\Delta(\partial C)}\neq 0%
	\quad\text{and}\quad%
	\beta_{2,1}=\frac{b_{1,1}}{\Delta(\widehat{\partial C})}=0.
\end{equation*}

Suppose that $f_{1},f_{2}\in\Z[x]$ such that $\deg f_{2}\leq \deg f_{1}$ and $\{(f_{1})^{T}_{d},(f_{2})^{T}_{d}\}$ is a basis of the integer kernel of $\partial C$. Then there exists a unimodular matrix $U=(u_{i,j})_{i=1,2;j=1,2}$ such that
\begin{equation*}
	U\cdot(\beta_{i,j})_{i=1,2;j=1,\ldots,d+1}=(f_{1},f_{2})_{d}.
\end{equation*}
As $C$ is nonsingular, $\deg f_{1}=d$. Therefore, $\lc(f_{1})=u_{1,1}\beta_{1,1}$ since $\beta_{2,1}=0$. If $\deg f_{2}<\deg f_{1}$, then $u_{2,1}=0$ since $\beta_{1,1}\neq 0$ and $\beta_{2,1}=0$. Thus, if $\deg f_{2}<\deg f_{1}$, then $u_{1,1}=\pm 1$ since $U$ is unimodular.
\end{proof}

Property~\eqref{polySizeBounds} of Theorem~\ref{thm:gp-to-polys} is now proved, completing the proof of the theorem:

\begin{lemma}\label{lem:polySizeBounds} Let $d\geq 2$ and $\vec{c}=[c_{2d-2},\ldots,c_{0}]$ be a geometric progression modulo $N$ such that $C=C(c_{2d-2},\ldots,c_{0})$ is nonsingular and $\gcd(\Delta(\vec{c}),N)=1$. If $f_{1},f_{2}\in\Z[x]$ such that $2\leq\deg f_{2}\leq\deg f_{1}$ and $\left\{(f_{1})^{T}_{d},(f_{2})^{T}_{d}\right\}$ is a basis of the integer kernel of $\partial C$, then
\begin{equation*}
  		\norm{\vec{c}/\Delta(\vec{c})}_{2,s^{-1}}%
  		\leq\frac{\left|\det C\right|^{d-2}}{\Delta(\partial C)^{d-1}}\norm{\vec{c}}_{2,s^{-1}}
  		\leq\left(s^{\frac{\deg f_{2}-d}{2}}\left|\sin\theta_{s}(f_{1},f_{2})\right|\norm{f_{1}}_{2,s}\norm{f_{2}}_{2,s}\right)^{d-1}
\end{equation*}
and
\begin{equation*}
	s^{\frac{\deg f_{2}-d}{2}}\left|\sin\theta_{s}(f_{1},f_{2})\right|\norm{f_{1}}_{2,s}\norm{f_{2}}_{2,s}%
	\leq\frac{1}{\Delta(\partial C)}\norm{\vec{c}}^{d-1}_{2,s^{-1}}%
	\leq\frac{1}{N^{d-2}}\norm{\vec{c}/\Delta(\vec{c})}^{d-1}_{2,s^{-1}}
\end{equation*}
for all $s>0$.
\end{lemma} 
\begin{proof} Let $d\geq 2$ and $\vec{c}=[c_{2d-2},\ldots,c_{0}]$ be a geometric progression modulo $N$ such that $C=C(c_{2d-2},\ldots,c_{0})$ is nonsingular and $\gcd(\Delta(\vec{c}),N)=1$. Suppose that $f_{1},f_{2}\in\Z[x]$ such that $2\leq\deg f_{2}\leq\deg f_{1}$ and $\left\{(f_{1})^{T}_{d},(f_{2})^{T}_{d}\right\}$ is a basis for the integer kernel of $\partial C$. Let $\vec{c}_{d}=\vec{c}_{d}(f_{1},f_{2})$. Then Lemma~\ref{lem:resFormula} implies that $\res(f_{1},f_{2})$ and $\Delta(\mathrm{S}_{d}(f_{1},f_{2}))$ are nonzero. Thus, $f_{1}$ and $f_{2}$ are coprime and Lemma~\ref{lem:gp-kernel} implies that $\vec{c}_{d}=\pm(\Delta(\mathrm{S}_{d}(f_{1},f_{2}))/\Delta(\vec{c}))\cdot\vec{c}$. Therefore, Lemma~\ref{lem:resFormula} implies that
\begin{equation}\label{eqn:cd-formula}
	\norm{\vec{c}_{d}}_{2,s^{-1}}=\frac{\Delta(\mathrm{S}_{d}(f_{1},f_{2}))}{\Delta(\vec{c})}\norm{\vec{c}}_{2,s^{-1}}=\frac{\left|\det C\right|^{d-2}}{\Delta(\partial C)^{d-1}}\norm{\vec{c}}_{2,s^{-1}}\quad\text{for all $s>0$}.
\end{equation}
As $\Delta(\mathrm{S}_{d}(f_{1},f_{2}))$ is a nonzero integer, the inequality $\norm{\vec{c}_{d}}_{2,s^{-1}}\geq\norm{\vec{c}/\Delta(\vec{c})}_{2,s^{-1}}$ holds for all $s>0$. Combining this inequality with \eqref{eqn:cd-formula} and the upper bound on $\norm{\vec{c}_{d}}_{2,s^{-1}}$ provided by Lemma~\ref{lem:ct-upper-bnd} yields the first set of inequalities stated in the corollary for all $s>0$.

Corollary~\ref{lem:ct-lower-bnd-td-k1}, Lemma~\ref{lem:resFormula} and \eqref{eqn:cd-formula} imply that
\begin{equation}\label{eqn:partial-bnd}
	s^{\frac{\deg f_{2}-d}{2}}\left|\sin\theta_{s}\right|\norm{f_{1}}_{2,s}\norm{f_{2}}_{2,s}%
	\leq\left|\frac{\Delta(\partial C)^{d}}{(\det C)^{d-1}}\right|^{d-2}\norm{\vec{c}_{d}}^{d-1}_{2,s^{-1}}%
	=\frac{\norm{\vec{c}}^{d-1}_{2,s^{-1}}}{\Delta(\partial C)},
\end{equation}
where $\theta_{s}=\theta_{s}(f_{1},f_{2})$, for all $s>0$. Let $r$ be the ratio of $\vec{c}$ modulo $N$. Then subtracting $r$ times row $i+1$ of $\partial C$ from row $i$ for $i=1,\ldots,d-2$ produces a matrix whose first $d-2$ rows contain multiples of $N$. As $\Delta(\vec{c})$ divides each entry of $\partial C$ and $\gcd(\Delta(\vec{c}),N)=1$, it follows that $\Delta(\vec{c})^{d-1}N^{d-2}$ divides $\Delta(\partial C)$. Thus, $\Delta(\vec{c})^{d-1}N^{d-2}\leq\Delta(\partial C)$ since $\Delta(\partial C)\neq 0$. Combining this inequality with \eqref{eqn:partial-bnd} completes the proof of the second set of inequalities stated in the corollary.
\end{proof}

In the proof of Lemma~\ref{lem:polySizeBounds}, the assumption that $\vec{c}$ is a geometric progression modulo $N$ such that $\gcd(\Delta(\vec{c}),N)=1$ is only used to prove the last inequality of the lemma. Setting $\vec{c}=(0,\ldots,0,1,0,\ldots,0)$ where the $1$ appears in the $d$th coordinate shows that the remaining inequalities cannot be improved by a constant factor without using the assumption. If $\deg f_{2}=d$, then \eqref{eqn:cd-formula} and Corollary~\ref{lem:ct-lower-bnd-k0} imply that
\begin{equation*}
	\frac{\left|\det C\right|^{d-2}}{\Delta(\partial C)^{d-1}}\norm{\vec{c}}_{2,s^{-1}}%
	\geq\left|\res(f_{1},f_{2})\right|^{1-1/d}%
	\quad\text{for all $s>0$.}
\end{equation*}
If the inequality is strict, then Lemma~\ref{lem:polySizeBounds} improves upon the lower bound on $\left|\sin\theta_{s}(f_{1},f_{2})\right|\norm{f_{1}}_{2,s}\norm{f_{2}}_{2,s}$ provided by Lemma~\ref{lem:resultantBound}.

\section*{Acknowledgements}

Part of this work was performed while the author was employed by The University of Queensland, Brisbane, Australia. The author is grateful to Cyril Bouvier for many helpful comments and discussions.

\bibliographystyle{amsplain}

\providecommand{\bysame}{\leavevmode\hbox to3em{\hrulefill}\thinspace}
\providecommand{\MR}{\relax\ifhmode\unskip\space\fi MR }
\providecommand{\MRhref}[2]{%
  \href{http://www.ams.org/mathscinet-getitem?mr=#1}{#2}
}
\providecommand{\href}[2]{#2}

\end{document}